\documentclass[twocolumn,pra,twocolumn,superscriptaddress]{revtex4}
\usepackage{color}
\usepackage{amsmath}
\usepackage{amssymb}
\usepackage{graphicx}
\usepackage{tikz}
\usepackage{xr}
\usepackage{float}
\usepackage{bbm}
 \usepackage{mathrsfs}
\usepackage{epsfig} 
\usepackage{verbatim}
\usepackage{array}
\usepackage{setspace}

\usepackage[unicode=true,
 bookmarks=true,bookmarksnumbered=false,bookmarksopen=false,
 breaklinks=false,pdfborder={0 0 1},backref=false,colorlinks=true]
 {hyperref}
\hypersetup{
 linkcolor=magenta, urlcolor=blue, citecolor=blue, pdfstartview={FitH}, hyperfootnotes=false, unicode=true}
 
\usepackage{color}
\newenvironment{proof}{\noindent\textbf{Proof\ }}{\hfill\rule{3mm}{3mm}}

\newtheorem{theorem}{Theorem}

\makeatletter
\@ifundefined{textcolor}{}
{%
 \definecolor{BLACK}{gray}{0}
 \definecolor{WHITE}{gray}{1}
 \definecolor{RED}{rgb}{1,0,0}
 \definecolor{GREEN}{rgb}{0,1,0}
 \definecolor{BLUE}{rgb}{0,0,1}
 \definecolor{CYAN}{cmyk}{1,0,0,0}
 \definecolor{MAGENTA}{cmyk}{0,1,0,0}
 \definecolor{YELLOW}{cmyk}{0,0,1,0}
}

\usepackage{amsfonts}\usepackage{tabularx}\usepackage{dcolumn}\usepackage{bm}\usepackage{graphicx}\usepackage{epstopdf}

\DeclareMathOperator{\tr}{tr}

\hypersetup{urlcolor=blue}

\makeatother

\newcolumntype{C}[1]{>{\centering\arraybackslash$}p{#1}<{$}}

\usepackage{algorithm}
\usepackage{algorithmic}

\begin{document}

\widetext

\title{Generic detection-based error-mitigation using quantum autoencoders}

\author{Xiao-Ming Zhang}
\affiliation{Department of Physics, City University of Hong Kong, Tat Chee Avenue, Kowloon, Hong Kong SAR, China}

\author{Weicheng Kong}
\affiliation{Origin Quantum Computing Company Limited, Hefei, Anhui, 230088, China}

\author{Muhammad Usman Farooq}
\affiliation{Department of Mathematics, City University of Hong Kong, Tat Chee Avenue, Kowloon, Hong Kong SAR, China}

\author{Man-Hong Yung}
\affiliation{Shenzhen Institute for Quantum Science and Engineering and Department of Physics, Southern University of Science and Technology, Shenzhen 518055, China}
\affiliation{Shenzhen Key Laboratory of Quantum Science and Engineering, Southern University of Science and Technology, Shenzhen, 518055, China}

\author{Guoping Guo}
\affiliation{CAS Key Laboratory of Quantum Information, University of Science and Technology of China, Hefei, Anhui, 230026, China}

\author{Xin Wang}
\email{x.wang@cityu.edu.hk}
\affiliation{Department of Physics, City University of Hong Kong, Tat Chee Avenue, Kowloon, Hong Kong SAR, China}
\affiliation{City University of Hong Kong Shenzhen Research Institute, Shenzhen, Guangdong 518057, China}

\begin{abstract} 

Efficient error-mitigation techniques demanding minimal resources is key to quantum information processing.
We propose a generic protocol to mitigate quantum errors using detection-based quantum autoencoders. In our protocol, the quantum data are compressed into a latent subspace while leaving errors outside, the latter of which is then removed by a measurement and post-selection. Compared to previously developed methods, our protocol on the one hand requires no extra qubits, and on the other hand has a near-optimal denoising power, in which under reasonable requirements all errors detected outside of the latent subspace can be removed, while those inside the subspace cannot be removed by any means.  Our detection-based quantum autoencoders are therefore particularly useful for near-term quantum devices in which controllable qubits are limited while noise reduction is important.
\end{abstract}
\maketitle

\section{Introduction}
 Mitigating errors are key to quantum information processing. Among the many techniques developed for this purpose, the concept of subspace is ubiquitous. In quantum error correction \cite{Nielsen.02}, certain stabilizers define a subspace where the quantum states are verified: corrupted states, detected by syndrome measurements as outside of the subspace, are corrected by recovery operations. Alternatively, quantum computation can be conducted in decoherence-free subspaces \cite{Altepeter.04,Xue.06,Fong.11,Friesen.17}, which are chosen to decouple completely from certain environmental noises and thereby protecting the desired operations.

It is expected that in the noisy intermediate-scale quantum (NISQ) era \cite{Preskill.18}, quantum algorithms \cite{Peruzzo.14,Kandala.17,Farhi.14,Farhi.16,McArdle.20} can be successfully run on about 50-100 qubits. Nevertheless, quantum error correction on the algorithms at this scale requires a much larger number of controllable qubits, posing a technological challenge. On the other hand, decoherence-free subspaces only exist for selective sources of error, and their power in mitigating noises is limited. This signifies the need for a more generic method to reduce different types of errors within a limited number of controllable qubits.  

In classical data processing, data compression plays a central role in noise reduction.
In principal component analysis \cite{Wold.87}, only the first few principal components of the data with the greatest signal-to-noise ratio are kept. Another  example is an autoencoder \cite{Vincent.08,Vincent.10}, a deep neural network with bottleneck layers in its center. Data fed to the network are processed by the bottleneck layers with substantially smaller numbers of neurons, before they are restored by the remaining layers to the original size. In both cases, error-mitigation is achieved by first compressing the data into a subspace while keeping most errors outside, and then recovering the data from the subspace.

These achievements have inspired applications of data compression to quantum information processing \cite{Romero.17,Wan.17,Zhao.19,Beer.20,Achache.20,Cao.20}. For example, a quantum autoencoder to compress data involving quantum states has been developed in \cite{Romero.17}. Its key ingredient is the support subspace of a set of density matrices $R$, defined as the vector space spanned by the eigenvectors with non-zero eigenvalues for all the density matrices in $R$,
that is typically smaller than the full Hilbert space. This quantum data compression is sometimes quite efficient, as has been demonstrated experimentally in circuit QED \cite{Pepper.19} and linear optics \cite{Huang.19}. Despite these exciting advances, it remains unclear how to denoise using quantum data compression.

Very recently, neural-network-based quantum autoencoders were proposed to denoise quantum data \cite{Bondarenko.19}. It was shown that for various types of states including GHZ-like states, W states, graph states, etc., different kinds of noises can be satisfactorily suppressed without fine-tuning the hyperparameters. Nevertheless, the input data have to entangle with hidden layers, requiring additional qubits which could be technologically challenging in NISQ devices.
Moreover, other techniques developed for near-term quantum devices, such as extrapolation \cite{Li.17,Temme.17}, constraining \cite{McClean.16,Ryabinkin.18}, and a stabilizer-like method \cite{SMcArdle.19} are typically specialized to certain problems and/or types of errors, which is not generic. 
These and other considerations necessitate a denoising method that is more generic and at the same time does not require additional qubits.

In this work, we develop a generic quantum error-mitigation protocol by combining a post-selection process with the autoencoder proposed in~\cite{Romero.17}.  On the one hand it requires no extra qubits and is therefore suitable for NISQ devices. On the other hand, the method is quite general: under reasonable requirements (see Sec.~\ref{subsec:req}), all errors detected outside of the support subspace can be removed, while those inside the support subspace cannot be removed by other means neither. The technique can be generically applied to different types of errors including global depolarization noise \cite{Boixo.18,Arute.19} and decay-type noise \cite{Ofek.16}.
The validity of our scheme is examined with W class states and others with local and global depolarization noises. For $4$-qubit W class states, we find that a two-stage compression method is more effective, where the input data are first compressed into a $3$-qubit subspace and then a $2$-qubit subspace. We believe that this error-mitigation method can be straightforwardly generalized to treat more complicated problems, and implemented on different experimental platforms.

\section{Quantum autoencoder} Our discussion follows the idea of the quantum autoencoder proposed in~\cite{Romero.17} that does not require extra qubits in the compression: when a given set of quantum data share a certain underlying structure, one may find a single unitary that can ``compress'' the data from the full Hilbert space into a subspace. 
More precisely, we consider a set of quantum states (represented by density matrices) $R=\{\rho_1,\rho_2\cdots\}$  in the Hilbert space $\mathcal{H}$. The underlying structure shared is manifested as the support, $\mathcal{S}$, of the set of states $R$, with $\dim\mathcal{S}<\dim\mathcal{H}$.
We then define two subspaces of $\mathcal{H}$: the latent subspace $\mathcal{L}$ and the junk subspace $\mathcal{J}$, where $\mathcal{H}=\mathcal{L}\oplus\mathcal{J}$. The latent subspace $\mathcal{L}$ is spanned by orthogonal bases $\{|L_1\rangle,|L_2\rangle,\cdots,|L_{\text{dim}\, \mathcal{L}}\rangle\}$, while the junk subspace $\mathcal{J}$ is spanned by the orthogonal bases $\{|J_1\rangle,|J_2\rangle,\cdots,|J_{\text{dim}\,\mathcal{J}}\rangle\}$.
The key to a quantum encoder is to find an encoding unitary $U_{\text{e}}$, such that for all $\rho\in R$,
 \begin{equation}
\sigma \equiv U_e\rho U_e^\dag=\sum_{i=1}^{\dim\mathcal{L}}p_i|\psi_i\rangle\langle\psi_i|,
 \end{equation}
 where
$|\psi_i\rangle \in \mathcal{L}$ (see~\cite{sm} for more details). 
\begin{figure} [!htbp]
\includegraphics[width=0.95\columnwidth]{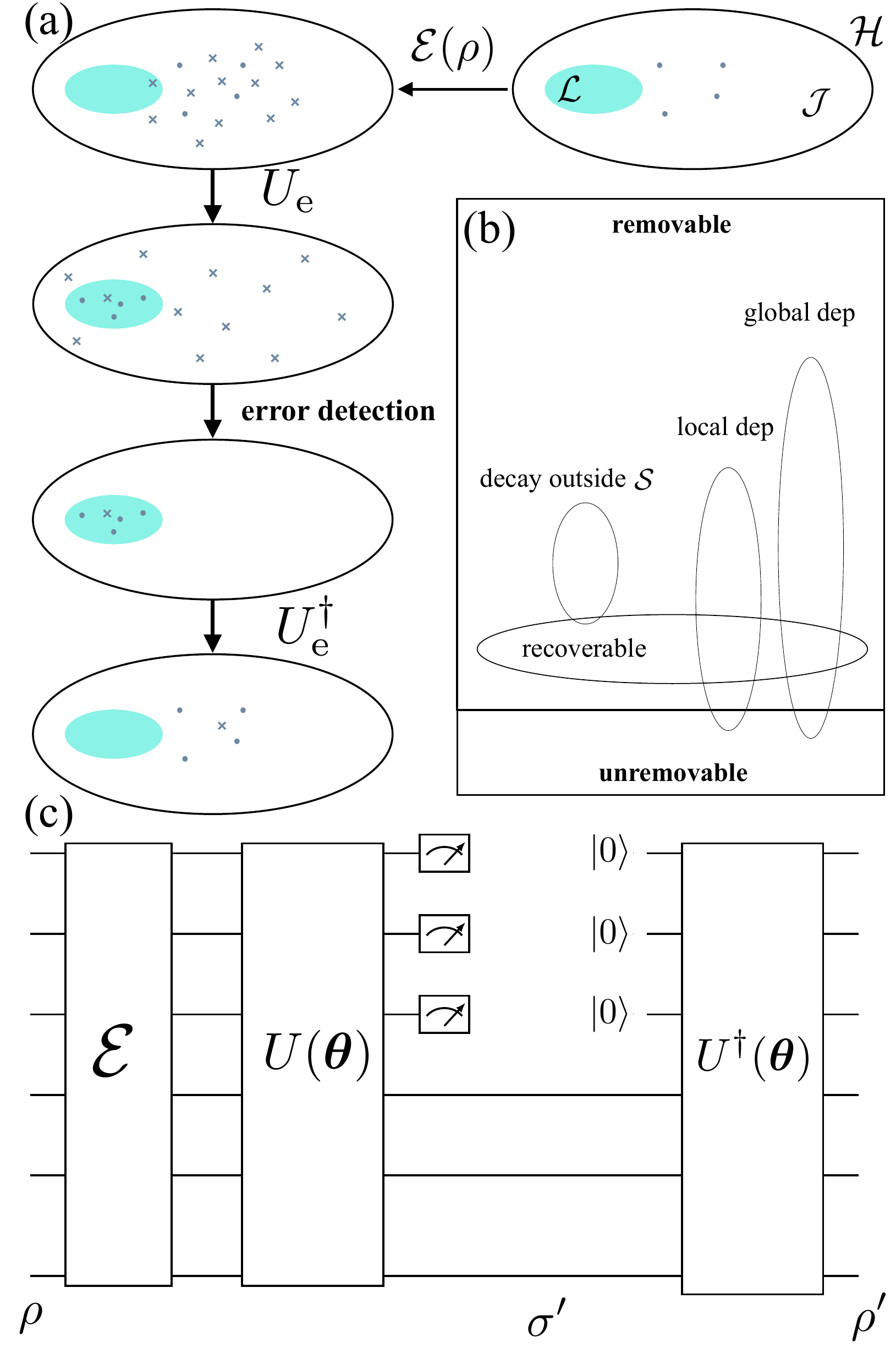}
\caption{(a) Sketch of the error detection and mitigation process. The large oval indicates the full Hilbert space $\mathcal{H}$, the green ellipse inside represents the latent subspace $\mathcal{L}$, while the remainder is the junk subspace $\mathcal{J}$. The noise effect is described by a quantum channel $\mathcal{E}(\cdot)$ as shown in Eq.~\eqref{eq:nc}. We use dots to represent the error-free term $(1-\varepsilon)\rho$, and crosses denote the error term $\varepsilon\rho^{\text{err}}$. $U_\text{e}$ transfers the error-free term  to $\mathcal{L}$, while most of the errors remain in $\mathcal{J}$. A measurement projects the state to the latent subspace, which detects and removes the errors in $\mathcal{J}$. Finally, $U_\text{e}^\dag$ is applied to recover the quantum data with error-mitigated. (b) Relations between different types of errors. Errors removable by our method include but are not limited to decay-type errors outside $\mathcal{S}$, and a large portion of local and global depolarization (denoted as ``dep") errors. All errors recoverable by quantum operations are also removable by our method. (c) The detection-based quantum autoencoder in the  circuit form.}
\label{fig:enc}
\end{figure}

 \begin{figure} [tbp]
\includegraphics[width=\columnwidth]{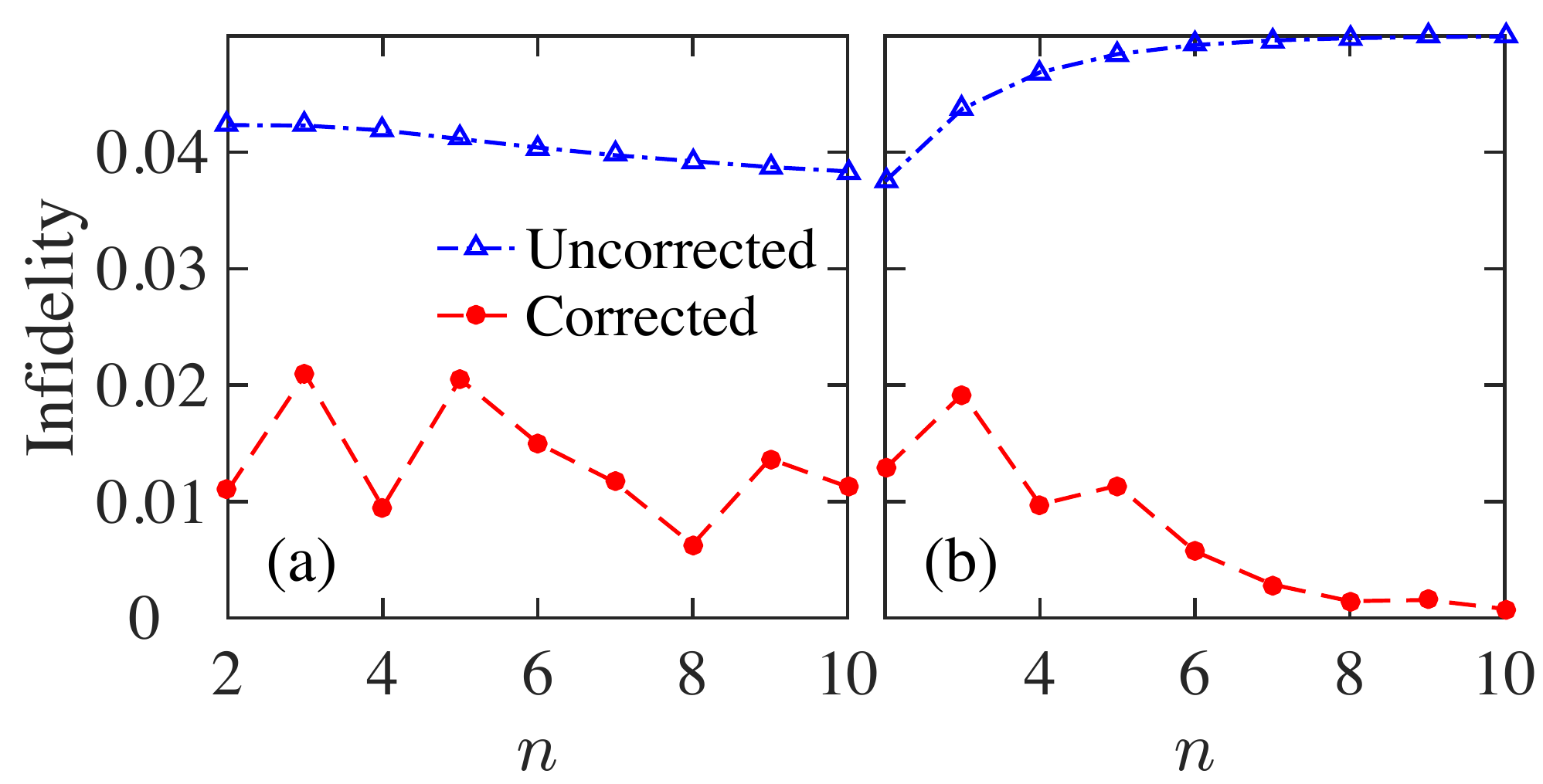}
\caption{ Performance of a detection-based quantum autoencoder for W class states with forms of input states known, under (a) local depolarization noise $\widetilde{\mathcal{E}}_{\text{lc}}(\cdot)$ and (b) global depolarization noise $\widetilde{\mathcal{E}}_{\text{gl}}(\cdot)$. We set $\varepsilon=0.05$.}
\label{fig:w}
\end{figure}

\begin{figure} [!htbp]
\includegraphics[width=0.8\columnwidth]{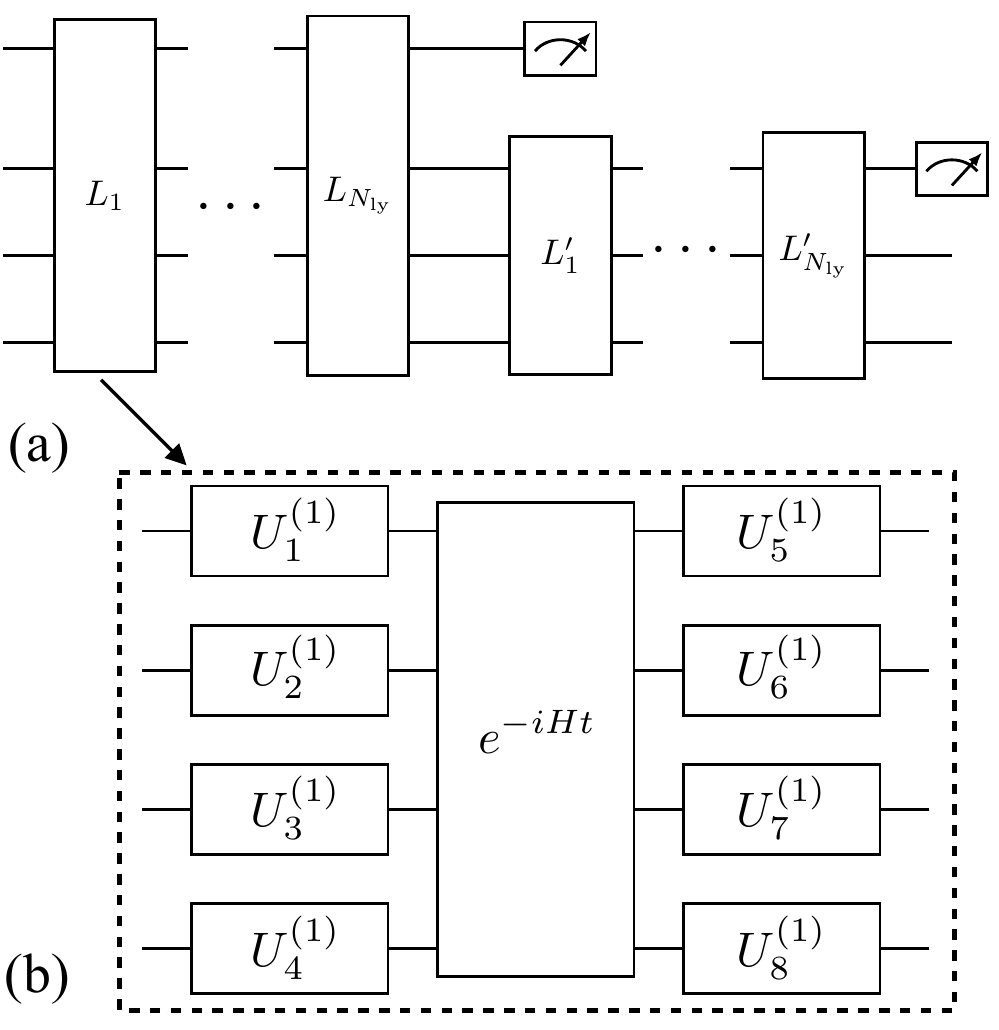}
\caption{(a) An example of programmable quantum autoencoders with two groups of $N_{\text{ly}}$-layer networks, which compress the $4$-qubit states to $2$-qubit states in two stages. (b) The circuit structure of each layer.}
\label{fig:c}
\end{figure}

\subsection{Optimization} When the form of $\mathcal{S}$ is known, it would be straightforward to determine $U_e$. However, $\mathcal{S}$ is unknown in most practical situations, and the quantum autoencoder should be obtained by optimization. Namely, one trains a programmable unitary such that the projection of data (but not noises) into $\mathcal{J}$ is minimized. Specifically, we define measurements $\{M_{L},M_{J} \}$ corresponding to the projection into the latent subspace or junk subspace respectively:
$M_L=\sum_{i=1}^{\dim \mathcal{L}}|L_i\rangle\langle L_i|,\; M_J=\sum_{i=1}^{\dim \mathcal{J}}|J_i\rangle\langle J_i|$. A programmable unitary $U(\bm{\theta})$ with a set of tunable parameters $\bm{\theta}$ is constructed, and applied to the input data $\rho_{\text{in}}$, giving $\sigma(\bm{\theta})\equiv U(\bm{\theta})\rho_{\text{in}} U(\bm{\theta})^{\dag}$. $\sigma(\bm{\theta})$ is expected to be fully in the latent subspace, so measurements of it in the junk subspace should give zero. We then define the cost function as
\begin{equation}
 C(\bm{\theta})=\tr \left[M_J\sigma({\bm{\theta}})\right].\label{eq:cost}
\end{equation} 
Minimizing $C(\bm{\theta})$ with respect to parameters $\bm{\theta}$ gives an approximated quantum autoencoder, $\widetilde{U}_{\text{e}}(\bm{\theta})$. In this work, we use the gradient descent method. In each iteration, the parameters are updated according to $\bm{\theta}\leftarrow\bm{\theta}-\gamma\nabla C(\bm{\theta})$, where $\gamma$ is the step size, and $\nabla C(\bm{\theta})$ the gradient of the cost function.
We note that the gradient descent could also follow cost functions other than Eq.~\eqref{eq:cost}. In Supplemental Material~\cite{sm}, we discuss the purity cost function as an alternative.

\subsection{Error mitigation} For noisy quantum data, the compressed state may leak out of the latent subspace, which actually provides a way for error detection. In the following, we use $\rho$ and $\sigma$ to represent the uncompressed and compressed quantum states, and a tilde indicates a noisy quantum state while a prime indicates that the error has been mitigated.

An ideal quantum state $\rho\in R$ deteriorated by noises can be generally expressed as
\begin{equation}
 \widetilde{\rho}=\mathcal{E}(\rho)=(1-\varepsilon)\rho+\varepsilon\rho^{\text{err}}. \label{eq:nc}
\end{equation}
Here the error term $\rho^{\text{err}}={\mathcal{E}}^{\text{err}}(\rho)=\sum_{i}A_i\rho A^\dag_i$, with $\sum_iA^\dag_iA_i=\mathbb{I}_{\dim\mathcal{H}}$ ($\mathbb{I}_N$ represents the $N$-dimensional identity). In other words, there is probability $\varepsilon\in[0,1]$ for the quantum data to be corrupted, which is determined by particular systems. We may define the infidelity between matrices $\rho_1,\rho_2$ as $\Delta(\rho_1,\rho_2)=1-\left(\text{Tr}\sqrt{\sqrt{\rho_1}\rho_2\sqrt{\rho_1}}\right)^2$. If $\rho$ is pure, we have $\Delta(\rho,\widetilde{\rho})=\varepsilon\Delta(\rho,\rho^{\text{err}})$.

Now we show that our detection-based quantum autoencoder can detect and dramatically reduce error (see Fig.~\ref{fig:enc}).
We begin with applying the encoding unitary $U_e$ to the noisy quantum state. The compressed state $\widetilde{\sigma}\equiv U_e\widetilde{\rho}\,U_e^\dag$ becomes $\widetilde{\sigma}=(1-\varepsilon)\sigma +\varepsilon\, U_e\rho^{\text{err}}U_e^\dag$. The key step is to perform the measurement $\{M_L,M_J\}$ for error detection. 
Since $\sigma$ is in the latent subspace, we have
\begin{align}\label{eq:m}
M_L\widetilde{\sigma}M_L^{\dag}&=(1-\varepsilon)\sigma+ \varepsilon M_LU_e\rho^{\text{err}}U_e^{\dag}M_L^{\dag}\notag\\
      &=(1-\varepsilon)\sigma+\varepsilon U_e \Lambda^{\text{err}}_s U_e^{\dag},
\end{align}
where $\Lambda_s^{\text{err}}\equiv M_S\rho^{\text{err}}M_S^{\dag}$ and $M_S\equiv U_e^{\dag}M_LU_e$. Assuming $\dim \mathcal{S}=\dim \mathcal{L}$, $M_S$ can always be rewritten as $M_{S}=\sum_{i=1}^{\dim\mathcal{L}}|S_i\rangle\langle S_i|$ for certain $|S_i\rangle\in \mathcal{S}$. Therefore, $\Lambda_s^{\text{err}}$ is the projection of $\rho^{\text{err}}$ to the support subspace $\mathcal{S}$. When $\varepsilon$ is small, according to Eq.~\eqref{eq:m}, with a high probability $p_s=1-\varepsilon\tr\left(\Lambda_s^{\text{err}}\right)$, the states are projected to the latent subspace $\mathcal{L}$, and become 
\begin{align}
\sigma'&=\frac{(1-\varepsilon)\sigma+\varepsilon U_e \Lambda^{\text{err}}_s U_e^{\dag}}{1-\varepsilon+\varepsilon\tr\left(\Lambda_s^{\text{err}}\right)  }\notag\\
&=[p_s+O(\varepsilon^2)]\sigma +[\varepsilon+O(\varepsilon^2)]\, U_e\Lambda_s^{\text{err}}\,U_e^\dag,
\end{align}
where higher order terms are determined by $\varepsilon$ and $\tr\left(\Lambda_s^{\text{err}}\right)$.
On the other hand, with a small probability $1-p_s$, the states are projected into $\mathcal{J}$ instead. In this case, the errors are detected, and the corresponding quantum data are discarded. 
Finally, the decoding unitary $U_e^\dag$ is applied to $\sigma'$ to obtain the error-mitigated state, $\rho'\equiv U_e^\dag\sigma'U_e$, which can be written as 
\begin{align}
\rho'=\left[1-\varepsilon\tr\left(\Lambda_s^{\text{err}}\right)+O(\varepsilon^2)\right]\rho +\left[\varepsilon+O(\varepsilon^2)\right]\Lambda_s^{\text{err}}.\label{eq:r_enc}
\end{align}
When $\rho$ is pure, the infidelity between $\rho'$ and $\rho$ becomes 
\begin{equation}
\Delta\left(\rho,\rho'\right)=\varepsilon\text{Tr}(\Lambda^{\text{err}}_s)\Delta(\rho,\Lambda^{\text{err}}_s/\text{Tr}(\Lambda^{\text{err}}_s)) +O(\varepsilon^2).\label{eq:I}
\end{equation}
We note that the detection-based autoencoder can effectively denoise when the noise $\widetilde{\mathcal{E}}(\cdot)$ mainly drives the quantum states out of the support subspace $\mathcal{S}$, i.e. $\text{Tr}(\Lambda^{\text{err}}_s)$ is small.

In~\cite{sm}, we study, as an example, the global depolarization noise: $\rho^{\text{err}}=\widetilde{\mathcal{E}}_\text{gl}(\rho)=\frac{1}{\dim\mathcal{H}}\mathbb{I}_{\dim\mathcal{H}}$, which is commonly adopted in large-scale superconducting circuits \cite{Boixo.18,Arute.19}. We have found that $\Delta(\rho,\rho')$, the infidelity for the corrected data, reduces exponentially with the number of qubits being measured during the compression, provided $\rho$ is pure. 

\begin{figure} [tbp]
\includegraphics[width=\columnwidth]{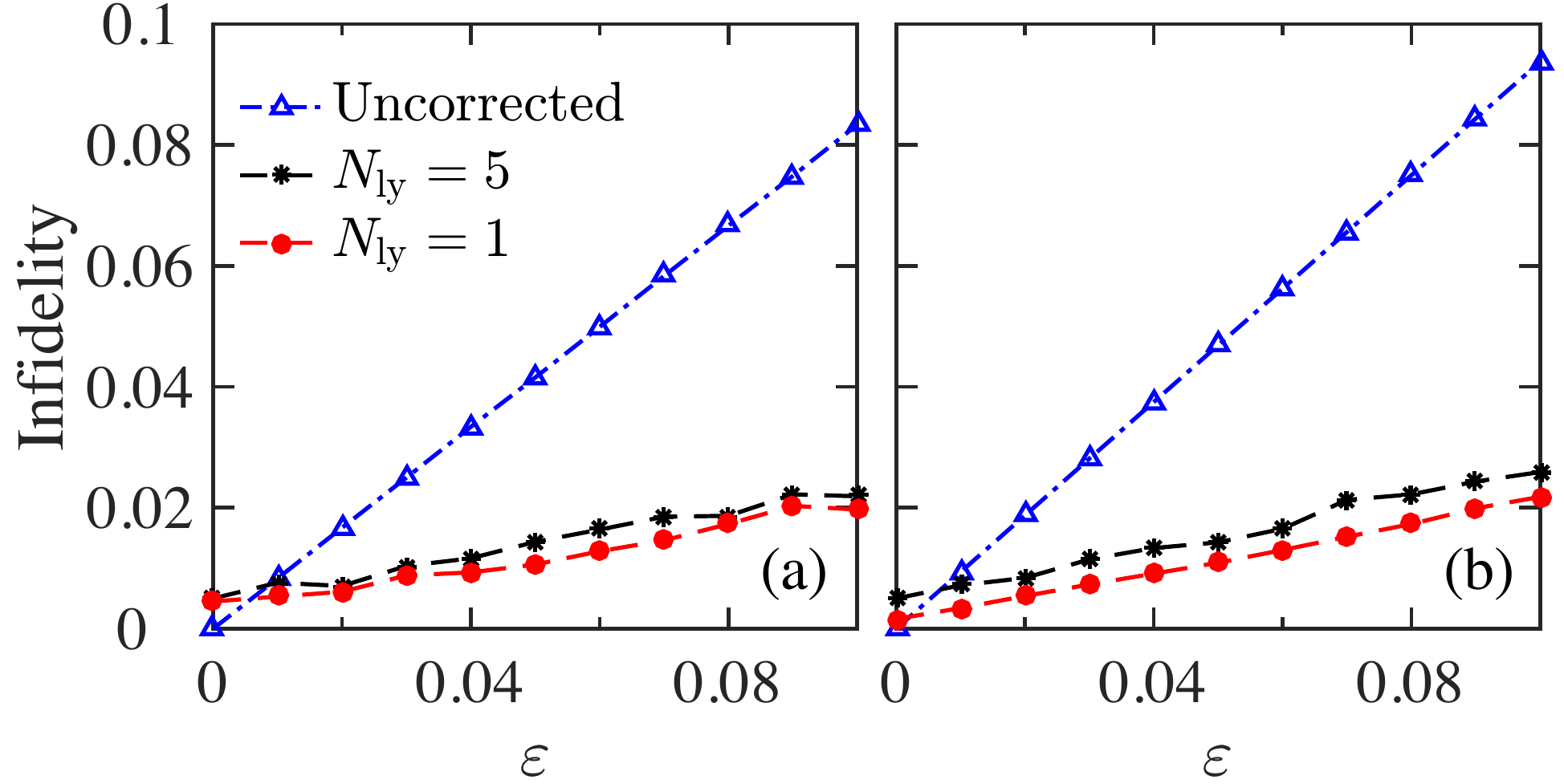}
\caption{ Performance of a detection-based quantum autoencoder for W class states with forms of input states unknown, under (a) local depolarization noise $\widetilde{\mathcal{E}}_{\text{lc}}(\cdot)$ and (b) global depolarization noise $\widetilde{\mathcal{E}}_{\text{gl}}(\cdot)$. We set $n=4$.}
\label{fig:train}
\end{figure}

\subsection{Applications} One may find applications of the detection-based quantum autoencoders whenever $\dim \mathcal{S}<\dim \mathcal{H}$. Here, we consider an example of $n$-qubit W class states which are frequently used to study entanglement \cite{Parashar.09} or spin preserving systems \cite{Christandl.04}:
\begin{equation}
|\psi_{\text{W}}\rangle=\alpha_1|10\cdots0\rangle+\alpha_2|01\cdots0\rangle+\cdots+\alpha_n|00\cdots1\rangle,\label{eq:w}
\end{equation}
for which $\dim \mathcal{H}=2^n$ and $\dim \mathcal{S}=n$. The compressed state is still an $n$-qubit state. However, the number of qubits required for encoding can be reduced if the latent subspace is chosen appropriately~\cite{Romero.17}.
We can then construct a latent subspace $\mathcal{L}$ spanned by the bases $\left\{|0\cdots0\, a_1 a_2\cdots a_m\rangle\right\}$, where $a_i=0,1$. Here, $m=\lceil\log_2n\rceil$ is the minimal number of qubits required to represent  $\mathcal{L}$. The corresponding measurements are $M_L=|0\rangle\langle0|^{\otimes n-m}\otimes \mathbb{I}_{2^{m}}$, and $M_J=\mathbb{I}_{2^n}-M_L$. In fact, $M_L$ corresponds to the projection of the first $(n-m)$ qubits to $|0\rangle$. These $(n-m)$ qubits are not used for encoding $\sigma$, so they can be considered as ``ancilla qubits" during the compression. However, they are still within the original data $\rho$, and no extra qubits are required.  The quantum circuit for the entire error-mitigation process is shown in Fig.~\ref{fig:enc}(c). 

We consider two noise models: the global depolarization noise $\widetilde{\mathcal{E}}_{\text{gl}}(\rho)$ mentioned above, and the local depolarization noise $\rho^{\text{err}}=\widetilde{\mathcal{E}}_{\text{lc}}(\rho)=\sum_{i=1}^{n}\sum_{\mu=x,y,z}\sigma_{i}^{\mu}\rho\,\sigma_{i}^{\mu}$. Here, $\sigma_{i}^{\mu}=\mathbb{I}_{2^{i-1}}\otimes\sigma^\mu\otimes \mathbb{I}_{2^{n-i}}$ are Pauli $X,Y$ or $Z$ operators acting on the $i$th qubit.

As a first verification of our detection-based autoencoder, we assume that the input states are known to be in the form of Eq.~\eqref{eq:w}. It is straightforward to construct the encoding unitary $U_\text{e}$ (see~\cite{sm}). Each data point represents an average over $1000$ runs with different input states, whose parameters $\alpha_i$ are randomly drawn from a normal distribution $\mathcal{N}(0,1)$ before the state is normalized.
 For different values of $n$ and noise models, the infidelities of $\rho'$ are much lower than the uncorrected data $\widetilde{\rho}$. In particular, for global depolarization noise [Fig.~\ref{fig:w}(b)], the infidelity for $\widetilde{\rho}$ increases with $n$ and converges to $0.1$.  For corrected states, however, the infidelity decreases exponentially with $n$, and finally converges to $0$. This trend is consistent with our analysis on the global depolarization noise in~\cite{sm}.

We then move on to a more complicated scenario in which the form of input states is unknown. We design a programmable circuit, and we minimize the corresponding cost function Eq.~\eqref{eq:cost}. As shown in Fig.~\ref{fig:c}, the programmable circuit has a layered structure, which is conceptually similar to those presented in \cite{Romero.17,Mitarai.18}, but the key difference is that we have two groups of $N_{\rm ly}$ layers that perform compression in different stages. Each layer contains a set of arbitrary single-qubit rotations and a global entangling unitary $e^{-iH\tau}$ (we set $\tau=1$), where 
\begin{align}
H=\sum_{i=1}^n\sum_{\mu=x,z}h_{i}^\mu\sigma_{i}^\mu+\sum_{i=1}^n\sum_{j>i}^n\sum_{\mu=x,y,z}J_{i,j}^\mu\sigma_{i}^\mu\sigma_{j}^\mu.\label{eq:Ham}
\end{align}
Here, $h_{i}^\mu$, $J_{i,j}^\mu$ are adjustable parameters. Unlike the previous example, the dimensionality of the latent space $\dim \mathcal{L}$ can only be found by trial and error. There are a total of $O(n^2N_{\text{ly}})$ parameters to be trained, and each trial takes time $O(N_{\text{ly}})$. So each interaction takes  time $O(n^2N_{\text{ly}}^2)$. Once an appropriate $\dim \mathcal{L}$ is found, we perform the compression in two stages, which turns out to be better than the single-stage method \cite{sm}. Taking $n=4$ as an example, we train the first group of $N_{\text{ly}}$ layers to compress the input states to the subspace spanned by $\{|0\,b_1b_2b_3\rangle\}$ with $b_i=0,1$ (driving the first qubit to $|0\rangle$) in the first stage. In the second stage, the remaining group of $N_{\text{ly}}$ layers further compress the states to the subspace spanned by $\{|0\,0\,a_1a_2\rangle\}$ with $a_i=0,1$ (driving the second qubit to $|0\rangle$). After the two stages, the quantum states have been successfully compressed to a latent subspace with $\dim \mathcal{L}=4$.

In our simulation, we take four $n=4$ states corresponding to $\{\alpha_1,\alpha_2,\alpha_3,\alpha_4\}=\{1,0,0,0\}$ , $\{0,1,0,0\}$ , $\{0,0,1,0\}$, and $\{0,0,0,1\}$ of Eq.~\eqref{eq:w} as training data, and the cost function Eq.~\eqref{eq:cost} is evaluated using the averaged results over these four states. The input of the quantum autoencoder is uncorrected states $\widetilde{\rho}$ described in Eq.~\eqref{eq:nc} under either the local or global depolarization noise.  The parameters of the circuit are trained with gradient descent until convergence.  Then the performance is tested with $1000$ quantum states of the form Eq.~\eqref{eq:w} with randomly generated $\alpha_i$. 

The comparison of infidelities for uncorrected states $\widetilde{\rho}$ and corrected states $\rho'$ is shown in Fig.~\ref{fig:train} (note that it is $C(\bm{\theta})$, rather than the infidelity, that is used as the cost function for training). When $\varepsilon=0$, $\rho'$ still has non-zero infidelities (0.0015 for $N_{\text{ly}}=1$ and $0.005$ for $N_{\text{ly}}=5$), which is an artifact of the training procedure. When $\varepsilon>0.01$, the infidelities for corrected states are much lower. In particular, for $\varepsilon=0.1$ and $N_{\text{ly}}=1$, the infidelities are reduced by $76\%$ and $77\%$  for local and global depolarization noises respectively. It is remarkable to note that a shallow circuit with $N_{\text{ly}}=1$ is already good enough for error-mitigation, and increasing $N_{\text{ly}}$ actually does not offer improvements (see \cite{sm} for more details). Our results indicate that the detection-based quantum autoencoder can learn the optimal compression methods even with noisy input data. 

Our protocol with detection-based quantum autoencoders is quite general and can be applied to a variety of problems. 
 In Supplemental Materials \cite{sm}, we also provide results under different practical situations including large noises, mixed states, as well as noisy circuits \cite{Johnson.17,Strikis.20}.

\section{Discussions on error-mitigation power}\label{sec:disconerror}

\subsection{Requirement on error-free data}\label{subsec:req}

In this section we show the generality of our method using detection-based quantum autoencoders, and we compare the error-mitigation power to neural-network-based ones. To facilitate the discussion and comparison, we impose the following requirement for ideal states unless otherwise specified: 
$\\$

\emph{The corrected state for an error-free state should also be error-free.}

$\\$
This requirement is well-satisfied by most standard error-mitigation techniques, including quantum error correction, dynamical decoupling and decoherence-free subspace.
\subsection{Generality}
Most existing error-mitigation methods for near-term quantum devices focus on specific types of errors. For example, constraining works well for errors breaking certain symmetries \cite{McClean.16,Ryabinkin.18}, while stabilizer-like methods \cite{SMcArdle.19} are suitable for depolarization error. Our method is general: it can remove all errors detected outside of the support subspace, which include but are not limited to the types of errors discussed above. More rigorously [see proof in~\cite{sm}], we have the following:

 \begin{theorem}\label{th:opt}
 If a quantum operation $\widehat{\mathcal{E}}(\cdot)$ satisfies the requirement stated in Sec.~\ref{subsec:req}, then $\widehat{\mathcal{E}}(\Lambda_s^{\text{err}})=\Lambda_s^{\text{err}}$.
 \end{theorem}
 Because the ideal detection-base quantum autoencoder can remove all errors outside $\mathcal{S}$, i.e.~$\Lambda_t^{\text{err}}$, Theorem~\ref{th:opt} implies that it is near optimal.  

\subsection{Comparision to a neural-network-based quantum autoencoder}
We now compare the detection-based and the neural-network-based autoencoders~\cite{Bondarenko.19}. The post-selection discarding states in the junk subspace enable us to mitigate a wide variety of error. For pure states, even in the worst-case scenario, the infidelity from our method is merely $O(\varepsilon^2)$ worse than \cite{Bondarenko.19}. More rigorously, if we define $U_{\text{nn}}$ as the neural-network-based quantum autoencoder, and $\rho'_{\text{nn}}$ as the output of it, we have the following theorem:
\begin{theorem}\label{th:com}
Given $U_{\text{nn}}$ satisfying the requirement stated in Sec.~\ref{subsec:req}, for an arbitrary $|\psi_a\rangle\in R$, we have $\Delta(|\psi_a\rangle\langle\psi_a|,\rho')\leqslant \Delta(|\psi_a\rangle\langle\psi_a|,\rho'_{\text{nn}})+O(\varepsilon^2).$
\end{theorem}

One of the major differences between neural-network-based autoencoders~\cite{Bondarenko.19} and our method is that for the former \cite{Bondarenko.19}, the error term is expected to be recovered to the correct state, while for our method, the error term is detected and removed. This is the reason why in extreme cases, the output infidelity of neural-network-based autoencoders is lower by $O(\varepsilon^2)$, especially when certain parts of the errors are neither recoverable nor detectable, and other parts of the errors are recoverable. This difference is, however, negligible when $\varepsilon\ll1$.

Moreover, there are errors such as decay-type errors that may not be treated directly in \cite{Bondarenko.19} but can be accommodated in our method (see~\cite{sm}).
 For data set $R=\{|i\rangle\langle i|\}$ with $i=1,2,\cdots,N$ and error term $\rho^{\text{err}}=|0\rangle\langle0|$ with $\langle i|j\rangle=\delta_{i,j}$, the encoding unitary is simply an identity, and it can be verified that the errors can be fully mitigated. But without a post-selection, no improvement is expected.
 
 The key point is that post-selection is non-trace-preserving, which makes our method general. In~\cite{sm}, we have also shown that if the approach in \cite{Bondarenko.19} is implemented in conjunction with measurement and post-selection, then the two methods would differ by at most $O(\varepsilon^2)$.

On the other hand, detection-based autoencoder has its own limitations. There is a risk that quantum data may be discarded if the states are projected to the junk subspace, and when $\varepsilon$ is large, the probability $\varepsilon\text{Tr}(\rho_s^{\text{err}})+O(\varepsilon^2)$ is not negligible. Moreover, the latent subspace should be carefully chosen to avoid over-compression, under-compression, and mismatch with the ansatz (detailed in~\cite{sm}).

Finally, we note that our detection-based method may be particularly useful when the qubit resources are limited, or when errors cannot be removed by mere trace-preserving mapping. On the other hand, when errors are large, or in situations where measurements could destroy quantum states (e.g.~in certain linear optical systems~\cite{Carolan.15}), the neural-network-based method \cite{Bondarenko.19} would be a better choice, provided sufficient qubit resources. These hurdles may be overcome when the neural-network-based auto-encoder is applied in conjunction with measurement and post-selection, which warrants further investigation.

\section{Conclusions and Outlook} 

Further improvement can be made to the detection-based autoencoder. Firstly, the junk subspace does not have to be predetermined. The form of $\mathcal{L}$ can alter during training provided that the dimension is fixed, which may potentially improve the performance. Secondly, we are using the simplest gradient-based optimization in finding the parameters $\bm{\theta}$, and more sophisticated algorithms such as gradient desent with momentum \cite{Rumelhart.85} and the Adam method \cite{Kingma.14} could also accelerate convergence. Other global optimization methods, such as simulated annealing \cite{Kirkpatrick.83} and reinforcement learning \cite{Niu.19,Xu.19,Zhang.19}, can help in preventing confinement of the optimization to a poor local minimum. Finally, optimization can also be made on the ansatz. In Eq.~\eqref{eq:Ham}, we have assumed a flexible Hamiltonian. When the Hamiltonian of a practical system is restricted, the unitary $e^{-iH\tau}$ should be decomposed into available gates with techniques such as Trotterization~\cite{Lloyd.96,Berry.07}. A simpler ansatz may reduce the number of gates required for decomposition, but is may also increase the number of layers in the neural network, while only a single layer is needed for Eq.~\eqref{eq:w} in our task. So there is a trade-off between the number of layers and the flexibility in each layer.

In summary, we have proposed a generic detection-based quantum autoencoder that can mitigate error without requiring additional qubits. By compressing the quantum data to a latent subspace, the error can be detected by a projection measurement. We believe that our protocol is particularly suitable for near-term NISQ devices, when the number of controllable qubits is not large while errors remain significant.

\section{Acknowledgements}  This work is supported by the Key-Area Research and Development Program of GuangDong Province  (Grant No. 2018B030326001), the National Natural Science Foundation of China (Grant Nos.~11874312, 11625419, 11875160, U1801661), the Research Grants Council of Hong Kong (Grant Nos.~CityU 11303617, CityU 11304018, CityU 11304920),
National Key Research and Development Program of China (Grant No. 2016YFA0301700), the Guangdong Innovative and Entrepreneurial Research Team Program (Grant No. 2016ZT06D348), Natural Science Foundation of Guangdong Province (Grant No. 2017B030308003), the Science,Technology and Innovation Commission of Shenzhen Municipality (Grant Nos.~JCYJ20170412152620376, JCYJ20170817105046702, KYTDPT20181011104202253), the Economy, Trade and Information Commission of Shenzhen Municipality (Grant No. 201901161512), and Guangdong Provincial Key Laboratory (Grant No. 2019B121203002).

%



\onecolumngrid
\vspace{1cm}

\begin{center}
{\bf\large Supplementary material}
\end{center}
\vspace{0.5cm}

\setcounter{secnumdepth}{3}  
\setcounter{equation}{0}
\setcounter{figure}{0}
\setcounter{table}{0}
\setcounter{section}{0}

\renewcommand{\theequation}{S-\arabic{equation}}
\renewcommand{\thefigure}{S\arabic{figure}}
\renewcommand{\thetable}{S-\Roman{table}}
\renewcommand\figurename{Supplementary Figure}
\renewcommand\tablename{Supplementary Table}
\newcommand\citetwo[2]{[S\citealp{#1}, S\citealp{#2}]}
\newcommand\citecite[2]{[\citealp{#1}, S\citealp{#2}]}

\newcolumntype{M}[1]{>{\centering\arraybackslash}m{#1}}
\newcolumntype{N}{@{}m{0pt}@{}}

\makeatletter \renewcommand\@biblabel[1]{[S#1]} \makeatother

\makeatletter \renewcommand\@biblabel[1]{[S#1]} \makeatother


\onecolumngrid
\section{Construction of the quantum autoencoder}\label{app:auto}
We consider a set of quantum states $R$ with support subspace $\mathcal{S}$. For simplicity, we define $N\equiv\dim\mathcal{H}$ and $M\equiv\dim\mathcal{S}$. All quantum states $\rho\in R$ can be written as 
  \begin{equation}
 \rho =\sum_{i=1}^{M}p_i|\psi_i\rangle\langle\psi_i|,
 \end{equation} 
where $|\psi_i\rangle\langle\psi_i| \in \mathcal{S}$. Suppose $\mathcal{S}$ is spanned by the orthogonal basis $\{|s_1\rangle,|s_2\rangle,\cdots,|s_{M}\rangle\}$, we define a latent subspace $\mathcal{L}$ with $\dim\mathcal{L}=M$, which is spanned by another set of orthogonal basis $\{|L_1\rangle,|L_2\rangle,\cdots,|L_{M}\rangle\}$. $U_\text{e}$ can be set as an arbitrary unitary satisfying
\begin{equation}
\langle L_i|U_\text{e}|s_i\rangle=1.\label{eq:ue}
\end{equation}
 According to the definition, all $|\psi_i\rangle\in\mathcal{S}$ can be written as $|\psi_i\rangle=\sum_{j=1}^{M} \alpha_{i,j} |s_j\rangle$ for certain values of $\alpha_{i,j}$. We define $|\phi_i\rangle\equiv U_{\text{e}}|\psi_i\rangle$. According to Eq.~\eqref{eq:ue}, $|\phi_i\rangle$ can always be written as the linear combination of $|L_i\rangle$, so $|\phi_i\rangle\in\mathcal{L}$. Therefore, the compressed state $\sigma=U_\text{e}\rho U_{\text{e}}^{\dag}$ can always be written as 
  \begin{equation}
 \sigma=\sum_{i=1}^{M}p_i|\phi_i\rangle\langle\phi_i|,
 \end{equation} 
 for certain $|\phi_i\rangle\in\mathcal{L}$. Therefore $U_\text{e}$ can be used as the encoding unitary for compression. Because the basis of $\mathcal{S}$ is not unique, $U_{\text{e}}$ is not unique neither.

\section{Global depolarization noise}\label{app:dp}
For global depolarization noise, we have $\widetilde{\mathcal{E}}(\rho)=\frac{1}{N}\mathbb{I}_{N}$. For an arbitrary pure state $\rho=|\psi_a\rangle\langle\psi_a|\in R$, the uncorrected quantum states [Eq.~\eqref{eq:nc} in the main text] and the corrected ones [Eq.~\eqref{eq:r_enc} in the main text] can be rewritten as 
\begin{equation}
\widetilde{\rho}=(1-\varepsilon)|\psi_a\rangle\langle\psi_a|+\varepsilon\frac{1}{N}\mathbb{I}_{N},\label{eq:g_0}
\end{equation}
and 
\begin{align}
\rho'=&\left[1-\varepsilon\frac{M}{N}+O(\varepsilon^2)\right]|\psi_a\rangle\langle\psi_a|\\
&+\left[\varepsilon+O(\varepsilon^2)\right]\frac{1}{N}\sum_{i=1}^{M}|S_i\rangle\langle S_i|.
\end{align}
So we have 
\begin{align}
\Delta(\rho,\widetilde{\rho})&=1-\left[(1-\varepsilon)+\frac{\varepsilon}{N} \langle\psi_a|\mathbb{I}_N|\psi_a\rangle\right] \notag\\
&=\frac{N-1}{N}\varepsilon
\end{align}
and similarly,
\begin{align}
\Delta(\rho,\rho')&=1-\left(1-\frac{M}{N}\varepsilon\right)-\frac{\varepsilon}{N}\langle\psi_a|\sum_{i=1}^{M}|s_i\rangle\langle s_i|\psi_a\rangle +O(\varepsilon^2)\notag\\
&=\frac{M-1}{N}\varepsilon+O(\varepsilon^2).
\end{align}
When $N\gg0$ and $\varepsilon\ll0$, we have 
\begin{equation}
\frac{\Delta\left(\rho,\rho'\right)}{\Delta\left(\rho,\widetilde{\rho}\right)}=\frac{M-1}{N-1}+O\left(\varepsilon^2\right)\approx M/N.
\end{equation}
The total number of qubits being measured during compression is at the order of $O(\log N/M)=O\left(\log \Delta(\rho,\widetilde{\rho})/\Delta(\rho,\rho')\right)$. In other words, the error of the corrected state $\rho'$ reduces exponentially with the number of qubits being measured. This is the result quoted in the main text.

\section{Quantum neural network and neural-network-based autoencoders}\label{app:nn}
In the framework proposed in Refs.~\cite{Bondarenko.19,Beer.20}, a Quantum neural network (QNN) includes a set of input qubits, hidden layer qubits, and output qubits. The initial state of the network is $\rho_\text{in}\otimes|0\rangle_{\text{hid,out}}\langle0|$, where $\rho_{\text{in}}$ is the initial quantum state of the input qubits, and $|0\rangle_{\text{hid,out}}\langle0|$ represents that all the hidden layer qubits and output qubits are initialized as $|0\rangle$. According to its definition, the output of the QNN can generally be described as 
\begin{equation}
\rho_\text{out}\equiv\text{Tr}_{\text{in,hid}}\left[ U_{\text{nn}}(\rho_\text{in}\otimes|0\rangle_{\text{hid,out}}\langle0| )U_{\text{nn}}^{\dag}\right],
\end{equation}
where $U_{\text{nn}}$ represents the QNN circuit 
(where ``nn'' stands for ``neural-network), and $\text{Tr}_{\text{in,hid}}$ is the partial trace over all input qubits and hidden layer qubits. Quantum autoencoder is a special type of QNN, whose number of output qubits is identical to the the input qubits~\cite{Bondarenko.19}. With noisy input data $\rho_{\text{in}}=\widetilde{\rho}$, one expects  the output states $\rho_{\text{out}}=\rho'_{\text{nn}}$ as close to the ideal quantum states $\rho$ as possible. 

\section{Comparison between detection-based autoencoders and neural-network-based ones}\label{app:comp}
Here we compare, in detail, the power in mitigating errors between our detection-based autoencoders and neural-network-based ones introduced in \cite{Beer.20}.
\subsection{Quantum data}
We assume the ideal states of quantum data are all pure states. Given a set of quantum data $R$, one can always find a total of $M$ linearly independent states and denote as $|\psi_1\rangle,|\psi_2\rangle,\cdots,|\psi_M\rangle$. The remaining states in $R$ can be denoted as $|\psi_a\rangle$ with $a\geqslant (M+1)$, which can always be written as the linear combination of $|\psi_a\rangle$ with $a\leqslant M$.  We denote the support subspace of $R$ as $\mathcal{S}$, which has dimension $\dim \mathcal{S}=M$. The orthogonal bases of $\mathcal{S}$ can be obtained by orthogonal decomposition of $|\psi_a\rangle$ with $1\leqslant a\leqslant M$ as follows: Firstly, we set $|S_1\rangle=|\psi_1\rangle$, and then we set
\begin{align}
|S_i\rangle=|S'_i\rangle/\sqrt{\langle S'_i|S'_i\rangle},
\end{align}
where
\begin{align}
|S'_i\rangle=|\psi_i\rangle-\sum_{a=1}^{i-1}|S_a\rangle\langle S_a|\psi_a\rangle,
\end{align}
for $i>1$. Then, $\{|S_1\rangle,|S_2\rangle,\cdots,|S_{M}\rangle\}$ forms a set of orthogonal bases of $\mathcal{S}$ (note that $|\psi_a\rangle\in\mathcal{S}$).

As described in the main text, for a given quantum data $\rho=|\psi_a\rangle\langle\psi_a|$, the deteriorated states can be generally expressed as 
\begin{align}
\widetilde{\rho}=\mathcal{E}(\rho)=(1-\varepsilon)\rho+\varepsilon\rho^{\text{err}},
\end{align}
 where $\rho^{\text{err}}$ is a density matrix representing the noise effect.

\subsection{Discussion and proof of Theorem~\ref{th:com}}
Under requirement described in Sec.~\ref{subsec:req} of the main text, we compare the best possible performance of the detection-based and neural-network-based autoencoders. Such requirement is already satisfied by our detection-based protocol, and for neural-network-based autoencoders, it is equivalent to
\begin{align}
\rho_\text{out}&=\text{Tr}_{\text{in,hid}}\left[ U_{\text{nn}}(|\psi_a\rangle_\text{in}\langle\psi_a|\otimes|0\rangle_{\text{hid,out}}\langle0| )U_{\text{nn}}^{\dag}\right]\notag\\
&=|\psi_a\rangle_{\text{out}}\langle\psi_a|,\label{eq:pr}
\end{align}
for an arbitrary $|\psi_a\rangle\in R$. Here, $|\psi_a\rangle_{\text{in}}\langle\psi_a|$ and $|\psi_a\rangle_{\text{out}}\langle\psi_a|$ represent pure states of the input qubits or output qubits with density matrix $|\psi_a\rangle\langle\psi_a|$. One must be cautious that this requirement may not be true in some extraordinary situations. There, one should simply expect the average fidelities for the corrected states be as high as possible. While a general comparison in these special scenarios is challenging, we  believe that our detection-based autoencoder still shows advantage based on Theorem~\ref{th:com} and an example that does not rely on the requirement in Sec.~\ref{subsec:req} of the main text (explained in Sec.~\ref{sec:exp}).

With uncorrected state $\widetilde{\rho}$ corresponding to a quantum data $\rho=|\psi_a\rangle\langle\psi_a|$, we denote the optimal corrected states of our detection-based and neural-network-based autoencoder as $\rho'$ and $\rho'_{\text{nn}}$ respectively. Theorem~\ref{th:com} states that in most cases (especially $\varepsilon$ is small), the infidelity of $\Delta(\psi_a\rangle\langle\psi_a|,\rho')$ is lower than $\Delta(|\psi_a\rangle\langle\psi_a|,\rho'_{\text{nn}})$, and in the worst-case scenario, $\Delta(\psi_a\rangle\langle\psi_a|,\rho')$ is no larger than $\Delta(|\psi_a\rangle\langle\psi_a|,\rho'_{\text{nn}})+O(\varepsilon^2)$. 

The proof of Theorem~\ref{th:com} is as follows:

\begin{proof}
According to Eq.~\eqref{eq:pr}, for ideal input states $|\psi_a\rangle$, we have
\begin{equation}
 U_{\text{nn}}|\psi_a\rangle_\text{in}\otimes|0\rangle_{\text{hid,out}}=|\phi_a\rangle_{\text{in,hid}}\otimes|\psi_a\rangle_{\text{out}},\label{eq:e0}
\end{equation}
for certain quantum states  $|\phi_a\rangle_{\text{in,hid}}$ of the input and hidden qubits. Recalling $|\psi_1\rangle=|S_1\rangle$, $|\psi_2\rangle=x|S_1\rangle+y|S_2\rangle$ ($x$,$y$ being nonzero complex numbers), and Eq.~\eqref{eq:e0}, we have
\begin{align}\label{eq:s_1_s_2}
 U_{\text{nn}}|S_1\rangle_\text{in}\otimes|0\rangle_{\text{hid,out}}=&|\phi_1\rangle_{\text{in,hid}} \otimes |S_1\rangle_{\text{out}},\notag\\
 U_{\text{nn}}|\psi_2\rangle_\text{in}\otimes|0\rangle_{\text{hid,out}}=&|\phi_2\rangle_{\text{in,hid}} \otimes (x|S_1\rangle+y|S_2\rangle_{\text{out}}),\notag\\  
U_{\text{nn}}|\psi_2\rangle_\text{in}\otimes|0\rangle_{\text{hid,out}}=&x|\phi_1\rangle_{\text{in,hid}}\otimes |S_1\rangle_{\text{out}}\notag\\
&+y\, U_\text{nn}|S_2\rangle_{\text{in}} \otimes |0\rangle_{\text{hid,out}}.
\end{align}
Combining Eq.~\eqref{eq:s_1_s_2}, we find that $U_\text{nn}|S_2\rangle_\text{in}\otimes|0\rangle_{\text{hid,out}}=|\phi_2\rangle_{\text{in,hid}}\otimes|S_2\rangle_\text{in}$ and $|\phi_2\rangle=|\phi_1\rangle$. With the same arguments for $|\psi_i\rangle$, it can be derived that $U_\text{nn}|S_i\rangle_\text{in}\otimes|0\rangle_{\text{hid,out}}=|\phi_i\rangle_{\text{in,hid}}\otimes|S_i\rangle_\text{in}$ and $|\phi_i\rangle=|\phi_1\rangle$ for arbitrary $i$. So we can define $|\phi\rangle\equiv|\phi_i\rangle$, and have
\begin{align}
U_\text{nn}|S_i\rangle_\text{in}\otimes|0\rangle_{\text{hid,out}}=|\phi\rangle_{\text{in,hid}}\otimes|S_i\rangle_\text{in}.\label{eq:phi}
\end{align}
The orthogonal basis of  $\mathcal{H}$ can be chosen as $\{|S_1\rangle,|S_2\rangle,\cdots,|S_M\rangle,|T_1\rangle,|T_2\rangle,\cdots, |T_{N-M}\rangle\}$ for certain $|T_j\rangle\in\mathcal{H}$ satisfying $\langle S_i|T_j\rangle=0$ and $\langle T_i|T_j\rangle=\delta_{i,j}$.
We also denote $|\Phi_j\rangle_{\text{in,hid,out}}$ as the full output quantum states with input $|T_j\rangle$:
\begin{equation}
 |\Phi_j\rangle_{\text{in,hid,out}}\equiv U_{\text{nn}}|T_j\rangle_\text{in}\otimes|0\rangle_{\text{hid,out}}.
\end{equation}
 For arbitrary $|T_j\rangle$, $|T_{j'}\rangle$ and $|S_i\rangle$, we have 
\begin{subequations}\label{eq:zr}
\begin{align}
&_{\text{in,hid,out}}\langle\Phi_j|\left(|\phi\rangle_{\text{in,hid}} \otimes |S_i\rangle_{\text{out}}\right)=0,\\
&_{\text{in,hid,out}}\langle\Phi_j|\Phi_{j'}\rangle_{\text{in,hid,out}}=\delta_{j,j'}.
\end{align}
\end{subequations}
 We separate $\rho^{\text{err}}$ into three terms
\begin{align}
\rho^{\text{err}}=\Lambda_s^{\text{err}}+\Lambda_{st}^{\text{err}}+\Lambda_{t}^{\text{err}},
\end{align}
where $\Lambda_s^{\text{err}}$ and $\Lambda_t^{\text{err}}$ are the projections of $\rho^{\text{err}}$ to $\mathcal{S}$ and the space spanned by $\{|T_i\rangle\}$, and $\Lambda_{st}^{\text{err}}$ includes the block off-diagonal terms. They can be written as
\begin{subequations}\label{eq:sub_lambda}
\begin{align}
&\Lambda_s^{\text{err}}=\sum_{i,i'=1}^{M}\alpha_{i,i'}|S_i\rangle\langle S_{i'}|,\\
&\Lambda_{st}^{\text{err}}=\sum_{i=1}^{M}\sum_{j=1}^{N-M}\left(\beta_{i,j}|S_i\rangle\langle T_j|+\beta^{*}_{i,j}|T_j\rangle\langle S_i|\right),\\
&\Lambda_{t}^{\text{err}}=\sum_{j,j'=1}^{N-M}\gamma_{j,j'}|T_j\rangle\langle T_{j'}|.
\end{align}
\end{subequations}
 With the input state $\widetilde{\rho}$, the full output quantum state of the quantum autoencoder circuit is 
\begin{align}
 \widetilde{\Gamma}\equiv& U_{\text{nn}}(\widetilde{\rho}\otimes|0\rangle_{\text{hid,out}}\langle0| )U_{\text{nn}}^{\dag}\notag\\
 =&(1-\varepsilon)|\phi\rangle_{\text{in,hid}}\langle\phi|\otimes|\psi_a\rangle\langle\psi_a| \\
 &+\varepsilon U_{\text{nn}}(\rho^{\text{err}}\otimes|0\rangle_{\text{in,hid}}\langle0| )U_{\text{nn}}^{\dag}\notag\\
=&(1-\varepsilon)|\phi\rangle_{\text{in,hid}}\langle\phi|\otimes|\psi_a\rangle\langle\psi_a|\\
&+\varepsilon(\Gamma^{\text{err}}_s+\Gamma^{\text{err}}_{st}+\Gamma^{\text{err}}_t),
\end{align}
where
\begin{subequations} 
\begin{align}
\Gamma^{\text{err}}_s&\equiv U_{\text{nn}}(\Lambda_{s}^{\text{err}}\otimes|0\rangle_{\text{hid,out}}\langle0| )U_{\text{nn}}^{\dag},\label{eq:gs_1}\\
\Gamma^{\text{err}}_{st}&\equiv U_{\text{nn}}(\Lambda_{st}^{\text{err}}\otimes|0\rangle_{\text{hid,out}}\langle0| )U_{\text{nn}}^{\dag},\\
\Gamma^{\text{err}}_t&\equiv U_{\text{nn}}(\Lambda_{t}^{\text{err}}\otimes|0\rangle_{\text{hid,out}}\langle0| )U_{\text{nn}}^{\dag}.
\end{align}
\end{subequations}
According to Eq.~\eqref{eq:phi}, $\Gamma^{\text{err}}_s$ can always be rewritten as $\Gamma^{\text{err}}_s=|\phi\rangle_{\text{in,hid}}\langle\phi|\otimes\Lambda_{s}^{\text{err}}$. Moreover, we have
\begin{subequations}
\begin{align}
\Gamma^{\text{err}}_{st}=&U_\text{nn}\left[\sum_{i,j=1}^{M,N-M}\beta_{i,j}(|S_i\rangle\otimes|0\rangle_{\text{hid,out}})(\langle T_j|\otimes\langle0|_{\text{hid,out}})\right]U_{\text{nn}}^\dag+\mathrm{H.c.}\\
=&\left[\sum_{i,j}\beta_{i,j}(|\phi\rangle_{\text{in,hid}}\otimes|S_i\rangle)(\langle\Phi_j|_{\text{in,hid,out}})\right]+\mathrm{H.c.}\\
\Gamma_t^{\text{err}}=&U_\text{nn}\left[\sum_{j,j'=1}^{N-M}\gamma_{j,j'}(|T_j\rangle\otimes|0\rangle_{\text{hid,out}})(\langle T_{j'}|\otimes\langle0|_{\text{hid,out}})\right]U_{\text{nn}}^\dag\notag\\
=&\sum_{j,j'=1}^{N-M}\gamma_{j,j'}|\Phi_j\rangle_{\text{in,hid,out}}\langle\Phi_{j'}|.
\end{align}
\end{subequations}
It can then be verified that $\widehat\Gamma_t^{\text{err}}\equiv\Gamma_t^{\text{err}}/\text{Tr}(\Gamma_t^{\text{err}})$ is a density matrix.

The final output of the neural-network-based autoencoder is the partial trace of $\widetilde{\Gamma}$ over all input and hidden layer qubits
\begin{align}
\rho'_{\text{nn}}&\equiv\text{Tr}_{\text{in,hid}}\left(\widetilde{\Gamma}\right)\notag\\
&=(1-\varepsilon)|\psi_a\rangle\langle\psi_a|+\varepsilon\left[\Lambda_s^{\text{err}} +\text{Tr}_{\text{in,hid}}\left(\Gamma^{\text{err}}_t\right)+\text{Tr}_{\text{in,hid}}\left(\Gamma^{\text{err}}_{st}\right)\right]\notag\\
&=(1-\varepsilon)|\psi_a\rangle\langle\psi_a|+\varepsilon\left[\Lambda_s^{\text{err}} +\text{Tr}(\Gamma_t^{\text{err}})\cdot\text{Tr}_{\text{in,hid}}\left(\widehat\rho^{\text{err}}_t\right)+\text{Tr}_{\text{in,hid}}\left(\Gamma^{\text{err}}_{st}\right)\right]\notag\\
&=(1-\varepsilon)|\psi_a\rangle\langle\psi_a|+\varepsilon\left[\text{Tr}(\Lambda_s^{\text{err}}) \cdot\widehat{\rho}_s^{\text{err}} +\text{Tr}(\Gamma_t^{\text{err}})\cdot\text{Tr}_{\text{in,hid}}\left(\widehat\rho^{\text{err}}_t\right)+\text{Tr}_{\text{in,hid}}\left(\Gamma^{\text{err}}_{st}\right)\right],\label{eq:out'}
\end{align}
where $\widehat\rho_s^{\text{err}}\equiv\Lambda_s^{\text{err}}/\text{Tr}\left(\Lambda_s^{\text{err}}\right)$ is a density matrix. According to Eq.~\eqref{eq:zr}, the last term of Eq.~\eqref{eq:out'} can be written as $\text{Tr}_{\text{in,hid}}\left(\Gamma^{\text{err}}_{st}\right)=\sum_{i=1}^{M}\sum_{j=1}^{N-M}\beta'_{i,j}|S_i\rangle\langle T_j|$ for certain values of $\beta'_{i,j}$, whose trace is zero. Therefore, we have $\text{Tr}(\Lambda_s^{\text{err}}) + \text{Tr}(\Gamma_t^{\text{err}})=1$ and  $\langle\psi| \text{Tr}_{\text{in,hid}}(\Gamma^{\text{err}}_{st}) |\psi\rangle=0$. The infidelity of the final output can be calculated as 
\begin{align}
\Delta(|\psi_a\rangle\langle\psi_a|,\rho'_{\text{nn}})&=1-\langle\psi_a|\rho'_{\text{nn}}|\psi_a\rangle\notag\\
&=\varepsilon [1 -  \text{Tr}(\Lambda_s^{\text{err}})\cdot\langle\psi_a|\widehat{\rho}_s^{\text{err}}|\psi_a\rangle-\text{Tr}(\Gamma_t^{\text{err}})\cdot\langle\psi_a|\text{Tr}_{\text{in,hid}}\left(\widehat\rho^{\text{err}}_t\right)|\psi_a\rangle ].\label{eq:inf_1}
\end{align}
Because $\widehat\Gamma_t^{\text{err}}$ is a density matrix, we have $\langle\psi_a|\text{Tr}_{\text{in,hid}}\left(\widehat\Gamma^{\text{err}}_t\right)|\psi_a\rangle\leqslant1$. So we have
\begin{align}
\Delta(|\psi\rangle\langle\psi|,\rho'_{\text{nn}})&\geqslant\varepsilon\left[1 -  \text{Tr}(\Lambda_s^{\text{err}})\cdot\langle\psi_a|\widehat{\rho}_s^{\text{err}}|\psi_a\rangle -\text{Tr}(\Gamma_t^{\text{err}})\right]\notag\\
&=\varepsilon\text{Tr}(\Lambda_s^{\text{err}})\left[1 -  \langle\psi_a|\widehat{\rho}_s^{\text{err}}|\psi_a\rangle \right]\notag\\
&= \varepsilon\text{Tr}(\Lambda_s^{\text{err}})\Delta(|\psi_a\rangle\langle\psi_a|,\widehat\rho^{\text{err}}_s).
 \end{align}
On the other hand, for our error-detection-based protocol, the corrected state is 
\begin{align}
\rho'&=\frac{(1-\varepsilon)|\phi_a\rangle\langle\phi_a|+\varepsilon \Lambda_s^{\text{err}}}{1-\varepsilon+\varepsilon\text{Tr}(\Lambda_s^{\text{err}})}\notag\\
&=\left[1-\varepsilon\text{Tr}(\Lambda^{\text{err}}_s)+O(\varepsilon^2)\right]|\phi_a\rangle\langle\phi_a|+\left[\varepsilon+O(\varepsilon^2)\right] \Lambda_s^{\text{err}},
\end{align}
and the corresponding infidelity is
\begin{align}
\Delta\left(|\phi_a\rangle\langle\phi_a|,\rho'\right)=\varepsilon\text{Tr}(\Lambda^{\text{err}}_s)\Delta(|\psi_a\rangle\langle\psi_a|,\widehat\rho^{\text{err}}_s)+O(\varepsilon^2).
\end{align} 
Therefore, we have 
\begin{align}
\Delta(|\psi_a\rangle\langle\psi_a|,\rho')\leqslant \Delta(|\psi_a\rangle\langle\psi_a|,\rho'_{\text{nn}})+O(\varepsilon^2).
\end{align} 
\end{proof}

\subsection{An example involving leakage error}\label{sec:exp}

The quantum data we consider is a set of pure states in the form $R=\left\{|i\rangle\langle i|\right\}$ with $i=1,2,\ldots$, and the leaked state is
\begin{equation}
\widetilde{\rho}=\mathcal{E}(|i\rangle\langle i|)=(1-\varepsilon)|i\rangle\langle i|+\varepsilon|0\rangle\langle0|.
\end{equation}
where $\langle i|j\rangle=\delta_{i,j}$. The support subspace of $R$ is simply spanned by all its elements. It can  be easily verified that in our detection-based autoencoder, the ideal encoding unitary is nothing but an identity operator $U_\text{e}=\mathbb{I}_{N}$, and the projection on $|0\rangle\langle0|$ can eliminate all error terms. So we have $\rho'=|i\rangle\langle i|$ and therefore $\Delta(|i\rangle\langle i|,\rho')=0$ for all $|i\rangle\langle i|\in R$. 

In the neural-network-based autoencoder, however, we have $\rho'_\text{nn}=\text{Tr}_{\text{in,hid}}\left(U_\text{nn} \widetilde\rho\otimes|0\rangle_\text{hid,out}\langle 0|  U_\text{nn}^\dag   \right)=(1-\varepsilon)\rho^{\text{ideal}}_{\text{out}}+\varepsilon\rho_\text{out}^{\text{err}}$ where $\rho^{\text{ideal}}_{\text{out}}=\text{Tr}_{\text{in,hid}}\left(U_\text{nn} \rho\otimes|0\rangle_\text{hid,out}\langle 0|  U_\text{nn}^\dag   \right)$ and $\rho^{\text{err}}_{\text{out}}=\text{Tr}_{\text{in,hid}}\left(U_\text{nn} |0\rangle\langle0|\otimes|0\rangle_\text{hid,out}\langle 0|  U_\text{nn}^\dag  \right)$. Note that $\rho^{\text{err}}_{\text{out}}$ is a constant density matrix for arbitrary input, and the average fidelity (over all states in $R$) between it and $|i\rangle$ is $\overline{\langle i|\rho^{\text{err}}_{\text{out}}|i\rangle}=1/N$. When $N$ is sufficiently large, we have $\lim_{N\rightarrow\infty}\overline{\langle i|\rho^{\text{err}}_{\text{out}}|i\rangle}=0$. In this case, the average infidelity over all $|i\rangle\langle i|\in R$ is 
\begin{equation}
\overline{\Delta}(|i\rangle\langle i|,\rho'_{\text{nn}})= 1-\overline{\langle i|(1-\varepsilon)\rho^{\text{ideal}}_{\text{out}}|i\rangle}\geqslant\varepsilon.\label{eq:no}
\end{equation}
 Because $\Delta(|i\rangle\langle i|,\widetilde{\rho})=\varepsilon$, Eq.~\eqref{eq:no} implies that the neural-network-based autoencoder cannot offer any improvement. Note that the result in Eq.~\eqref{eq:no} does not depend on the the requirement stated in Sec.~\ref{subsec:req}, i.e. Eq.~\eqref{eq:pr}.
 
 We also performed numerical simulations on the autoencoder that mitigates the leakage error, and results are shown in Fig.~\ref{fig:leak}. We set $N=4$,  choose the ansatz as 
 \begin{equation}
 U_\text{e}=\prod_{i,j=0}^{4} e^{i\alpha_{i,j} |i\rangle\langle j|},
 \end{equation}
 and set the projection to the latent subspace $C_p=\text{Tr}\left(M_{L}\sigma\right)$ as the cost function. The projection operator to the latent subspace $M_L$ is chosen as $\mathbb{I}_{5}-|0\rangle\langle0|$ (corresponding to $U_\text{e}=\mathbb{I}_5$) for Fig.~\ref{fig:leak}(a) and $\mathbb{I}_{5}-|4\rangle\langle4|$ for Fig.~\ref{fig:leak}(b). Other parameters are the same as those used for $W$ class states. As can be seen, the autoencoder can mitigate almost all errors after training for different projection operators.

However, we note that the performance of the autoencoder depends on the choice of measurement and variational ansatz. Error-mitigation may fail when $U_\text{e}$ cannot separate the ideal states and error terms into the latent and junk subspaces. For example, if one chooses the ansatz 
  \begin{equation}
 U^{\text{mismatch}}_\text{e}=\prod_{i,j=1}^{4} e^{i\alpha_{i,j} |i\rangle\langle j|}
 \end{equation} 
and latent subspace 
  \begin{equation}
M^{\text{mismatch}}_\text{L}=\mathbb{I}_5-|4\rangle\langle4|,
 \end{equation}  
any improvement on the infidelity is impossible. This is because the error term $|0\rangle\langle0|$ is decoupled from other levels, i.e.~for all quantum data $|i\rangle\langle i|$ with $i\neq0$, $U^{\text{mismatch}}_\text{e}$ is incapable to transform the error to the junk subspace, therefore the fidelity $F<1-\varepsilon$. For a more general discussion on measurement, see Sec.~\ref{sec:m}.
\begin{figure}[t]
\includegraphics[width=0.65\columnwidth]{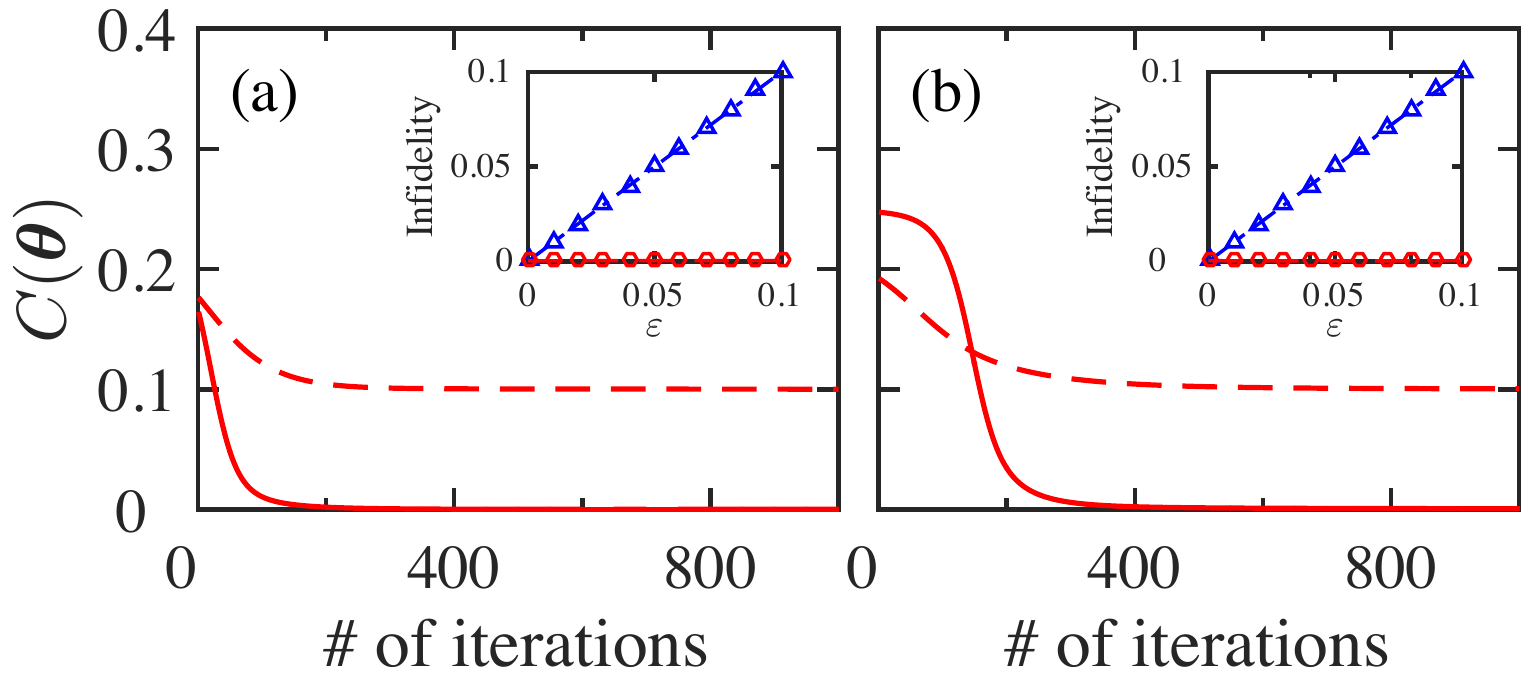}
\caption{Main panels: training curves of an autoencoder correcting leakage error for (a) $M_L=\mathbb{I}_{5}-|0\rangle\langle0|$ and (b) $M_L=\mathbb{I}_{5}-|4\rangle\langle4|$. Solid lines: $\varepsilon=0$, dashed lines: $\varepsilon=0.1$. Insets: infidelity versus error rate $\varepsilon$. Red circles: error-mitigated data, blue triangles: uncorrected data $\tilde{\rho}$.   }
\label{fig:leak}\end{figure}

\subsection{Neural-network-based autoencoder with post-selection}\label{app:nn+pos}
In this section we show that together with post-selection, which is the key of our work,  a neural-network-based autoencoder can perform comparably to our method.  
With measurement and post-selection applied within the procedure of neural-network-based autoencoder, the output state can generally be described as
 \begin{equation}
\rho_\text{out,1}=\Lambda_\text{out,1}/\text{Tr}(\Lambda_\text{out,1}),\label{eq:dp1}
\end{equation}
where
 \begin{align}
&\Lambda_\text{out,1}=\notag\\
&\text{Tr}_{\text{in,hid}}\left[ V_{\text{nn}} M_{\text{nn}} U_{\text{nn}}(\rho_\text{in}\otimes|0\rangle_{\text{hid,out}}\langle0| )U_{\text{nn}}^{\dag}M_{\text{nn}}^{\dag}V_{\text{nn}}^{\dag} \right].\label{eq:dp2}
\end{align}
Here $M_{\text{nn}}M_{\text{nn}}^\dag\leqslant\mathcal{I}_{\text{in,hid,out}}$ is a projection operator, and $U_{\text{nn}}, V_{\text{nn}}$ are unitaries. On the other hand, the output of the shallow circuits introduced in the main text of our paper is described as 
 \begin{equation}
\rho_\text{out,2}=\Lambda_\text{out,2}/\text{Tr}(\Lambda_\text{out,2}),\label{eq:sh1}
\end{equation}
where
\begin{equation}
\Lambda_{\text{out,2}}= U_{\text{e}}^{\dag}M U_{\text{e}}\rho_\text{in}U_{\text{e}}^{\dag}M^{\dag}U_{\text{e}}.\label{eq:sh2}
\end{equation}
 The network described in Eq.~\eqref{eq:dp1} and Eq.~\eqref{eq:dp2} has a comparable error-mitigation power to the shallow one [Eqs.~\eqref{eq:sh1} and \eqref{eq:sh2}], as it can always represent the shallow circuit:
\begin{theorem}\label{th:xxxx}
Given a shallow autoencoder in Eq.~\eqref{eq:sh1} and Eq.~\eqref{eq:sh2}, there always exists a neural-network-based autoencoder described in Eq.~\eqref{eq:dp1} and Eq.~\eqref{eq:dp2}, such that $\rho_\emph{out,1}=\rho_\emph{out,2}$.
\end{theorem}

\begin{proof}
We let $M_{\text{nn}}=M\otimes \mathbb{I}_{\text{in,hid}}$, $U_{\text{nn}}=U_{\text{e}}\otimes \mathbb{I}_{\text{in,hid}}$ and $V_{\text{nn}}=\left(U^{\dag}_{\text{e}}\otimes \mathbb{I}_{\text{in,hid}}\right)\text{SWAP}_{\text{in,out}}$, where $\text{SWAP}_{\text{in,out}}$ is the swap gates over each pair of qubits in input and output layers.  It is straightforward to show that $\rho_{\text{out,2}}=\rho_{\text{out,1}}$.
\end{proof}

\section{Discussion and proof of Theorem~\ref{th:opt}}\label{app:opt}

Because the neural-network-based quantum autoencoder is trace preserving, and can represent arbitrary unitaries of the system containing input, hidden layer and output qubit systems, it can in principle represent an arbitrary complete-positive-trace-preserving (CPTP) map acting on the input data $|\psi_a\rangle$. In the proof of Theorem~\ref{th:com}, one can observe that the error term in the support subspace, $\Lambda_s^{\text{err}}$, is always unchanged after applying the neural-network based quantum autoencoder. Therefore, $\Lambda_s^{\text{err}}$ is also unchanged under an arbitrary CPTP map satisfying the requirement described in Sec.~\ref{subsec:req} of the main text. Theorem~\ref{th:opt} states that the above argument is still true for general non-trace preserving quantum operations $\widehat{\mathcal{E}}(\cdot)$ in the following form:
 \begin{equation}
\widehat{\mathcal{E}}(\rho)=\sum_mE_m\rho E_m^{\dag},
\end{equation}
with $\sum_mE_mE_m^{\dag}\leqslant \mathbb{I}_N$, where $N$ is the dimension of $\rho$. The requirement described in Sec.~\ref{subsec:req} of the main text is equivalent to 
 \begin{equation}
\widehat{\mathcal{E}}(|\psi_a\rangle\langle\psi_a|)=|\psi_a\rangle\langle\psi_a|,\label{eq:npr}
\end{equation}
 for all $|\psi_a\rangle\langle\psi_a|\in R$. 
 
 The proof of Theorem~\ref{th:opt} is as follows:
 
\begin{proof}

 We introduce an environment $E$ with orthogonal basis $\{|e_0\rangle,|e_1\rangle, \cdots\}$, which is initialized in state $|e_0\rangle$. An arbitrary quantum operation can always be modelled by a unitary $U$ acting on the system-environment space and a projection operation $P$ acting on environment \cite{Nielsen.02}:
\begin{equation}
\widehat{\mathcal{E}}(\rho)=\text{Tr}_\text{E}\left[PU\left(\rho\otimes|e_0\rangle\langle e_0|\right)U^{\dag}P\right],
\end{equation}
where $\rho$ is the initial state of the system, and $\text{Tr}_\text{E}$ is the partial trace over $E$. According to Eq.~\eqref{eq:npr}, for arbitrary $|\psi_a\rangle\in R$, we have 
\begin{equation}
\text{Tr}_\text{E}\left[PU\left(|\psi_a\rangle\langle\psi_a|\otimes|0\rangle\langle0|\right)U^{\dag}P\right]=|\psi_a\rangle\langle\psi_a|,
\end{equation}
which means that 
\begin{equation}
PU\left(|\psi_a\rangle\langle\psi_a|\otimes|0\rangle\langle0|\right)U^{\dag}P=|\psi_a\rangle\langle\psi_a|\otimes\rho_{\text{env},a},\label{eq:pu}
\end{equation}
where $\rho_{\text{env},a}$ is a density matrix describing the environment satisfying $\text{Tr}[\rho_{\text{env},a}]=1$. From Eq.~\eqref{eq:pu}, we have 
\begin{equation}
U\left(|\psi_{a}\rangle\langle\psi_{a}|\otimes|0\rangle\langle0|\right)U^{\dag}=|\psi_a\rangle\langle\psi_a|\otimes|\psi _{\text{env},a}\rangle\langle\psi_{\text{env},a}|,\label{eq:up}
\end{equation}
where $|\psi _{\text{env},a}\rangle$ is a normalized state of environment. 
With the same argument as Eqs.~\eqref{eq:e0}-\eqref{eq:phi}, it can be derived that for arbitrary $|S_i\rangle$, 
\begin{equation}
U(|S_i\rangle\langle S_i|\otimes|0\rangle\langle0|)U^{\dag}=|S_i\rangle\langle S_i|\otimes|\psi _{\text{env}}\rangle\langle\psi _{\text{env}}|,\label{eq:us}
\end{equation}
is satisfied, where $|\psi _{\text{env},a}\rangle=|\psi _{\text{env}}\rangle$ for all $a$. Eq.~\eqref{eq:us} is equivalent to $U|S_i\rangle\otimes|0\rangle=|S_i\rangle\otimes|\psi _{\text{env}}\rangle$. Because $\Lambda_s^{\text{err}}$ is in the support subspace and can be represented by basis $\{|S_i\rangle\}$, we have
\begin{equation}
U(\Lambda_s^{\text{err}}\otimes|e_0\rangle\langle e_0|)U^{\dag}=\Lambda_s^{\text{err}}\otimes|\psi_\text{env}\rangle\langle \psi_\text{env}|.
\end{equation}
According to Eq.~\eqref{eq:pu}-\eqref{eq:up}, and notice that $|\psi _{\text{env},a}\rangle=|\psi _{\text{env}}\rangle$, we have
\begin{equation}
PU(\Lambda_s^{\text{err}}\otimes|e_0\rangle\langle e_0|)U^{\dag}P=\Lambda_s^{\text{err}}\otimes\rho_{\text{env}}
\end{equation}
with $\rho_{\text{env},a}=\rho_{\text{env}}$ for all $a$. Therefore, we have
\begin{equation}
\widehat{\mathcal{E}}\left(\Lambda_s^{\text{err}}\right) = \text{Tr}_E\left[PU(\Lambda_s^{\text{err}}\otimes|e_0\rangle\langle e_0|)U^{\dag}P\right]=\Lambda_s^{\text{err}}.
\end{equation}
\end{proof}

\section{Further discussion on measurement}\label{sec:m}
The performance of our detection-based quantum autoencoder depends crucially on the choice of measurements (the latent subspace) and the ansatz. A judicious choice can be made by taking into account of three considerations listed below.

Firstly, the measurement should be chosen such that the variational ansatz is easy to be trained to compress most, if not all, input states to the latent subspace. In Fig.~\ref{fig:mismatch}(a), we consider the task of compressing $4$-qubit $W$-class states into an 8-dimensional latent subspace with different measurements $M_L$. We compare the measurement $M_L=|0\rangle\langle0|\otimes \mathbb{I}_8$ as chosen in the main text (result shown as the dashed line) to 10 instances of $M_L$, which are constructed by $8$ randomly chosen orthogonal basis (shown as crosses). While the infidelities for some randomly chosen $M_L$ are comparable to results from $M_L=|0\rangle\langle0|\otimes \mathbb{I}_8$, others have higher infidelities. Here we note that in order to make a better choice of measurement, an effective method is to parametrize the measurement operator, and include the optimization of those parameters in the training process.

Secondly, the dimensionality of the latent subspace associated to the measurement should be as close as possible to that of the support for the quantum data set. On one hand, if the latent subspace is larger (under-compression), the latent subspace could include more errors so as to weaken the error-mitigation power of the autoencoder. On the other hand, if the latent subspace is smaller (over-compression), the information encoded in the original data could be lost. An example showing results of under-compression and over-compression has been given in Fig.~\ref{fig:mismatch}(b), where one can see that the infidelities for both under-compression and over-compression cases are higher than the case in which the latent subspace is of the right size. This challenge can be overcome by using multi-stage training method as mentioned in the main text. In each stage of training, the data are compressed to the latent subspace that is smaller than the one in the previous stage. A stage of compression is accepted only if the performance is improved.

\begin{figure}[t]
\includegraphics[width=0.65\columnwidth]{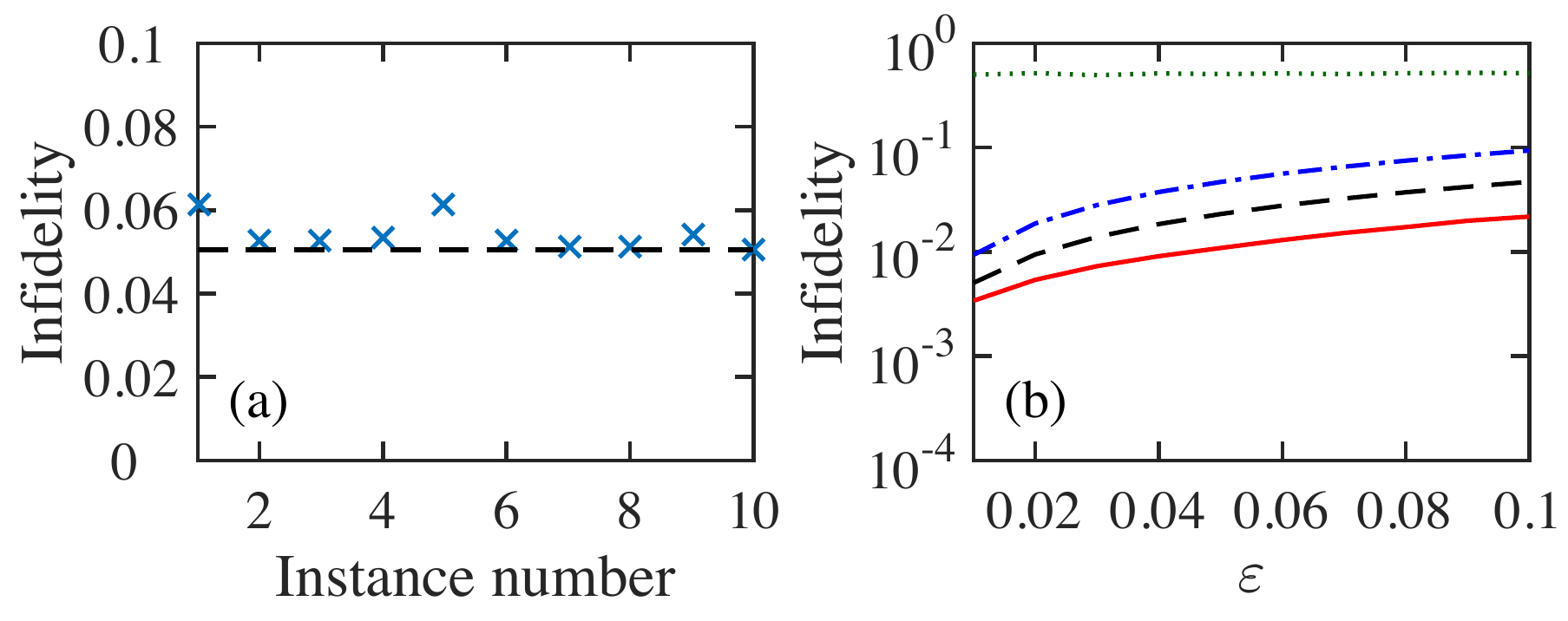}
\caption{Comparison of the error-mitigation performances using different projection measurements $M_L$, for $4$-qubit $W$-class states under global depolarization noise. (a) The crosses represent results from randomly chosen four dimensional latent subspaces, while the dashed line represents result from $M_L=|0\rangle\langle0|^{\otimes 2}\otimes \mathbb{I}_{4}$ as used in the main text. (b) Green dotted line represents $M_L=|0\rangle\langle0|^{\otimes 3}\otimes \mathbb{I}_{2}$ (over-compression), blue dash-dotted line the uncorrected quantum data, black dashed line $M_L=|0\rangle\langle0|^{\otimes 1}\otimes \mathbb{I}_{8}$ (under-compression), and red solid line  $M_L=|0\rangle\langle0|^{\otimes 2}\otimes \mathbb{I}_{4}$.  }\label{fig:mismatch}
\end{figure}

Thirdly, the measurement chosen should be easy to implement. For a qubit system, we have chosen $M_L$ such that it is the projection of the first $n-m$ qubits to  $|0\rangle$, and the compressed states is simply the tensor product of $|0\rangle\langle0|^{n-m}$ (ancilla qubit) and the quantum state of the remaining qubits. However, we note that there exist systems for which the compressed state can not be written in the form of a tensor product. An example is the continuous variable system~\cite{Weedbrook.12}. The information may be encoded in the frequency of a single photon with infinite dimension. In this case, one may define the latent subspace as the frequency lower than a certain value, and $M_L$ can be realized directly with a single long-pass filter.

\section{Supplementary results for W class states}\label{app:cm0}

\subsection{Training curves}
Examples of the training curves of the cost function $C(\bm{\theta})$ for three situations (ideal, global and local depolarization noises) are shown in Supplementary Fig.~\ref{fig:cv}. In all cases, $C(\bm{\theta})$ decreases smoothly during the training, and converges after about $2\times10^4$ iterations. We note that  under noise, $C(\bm{\theta})$ cannot reach zero. Instead, the lowest possible value is $\varepsilon \text{Tr}\left[M_J\Lambda^\text{err}\right]$, in which case one transforms the term $\Lambda_{s}^\text{err}$ fully into the junk subspace. To avoid confusion, $C(\bm{\theta})$ is not directly related to the fidelity of corrected states, because the fidelity is mainly determined by the remaining errors inside the support subspace $\Lambda_{t}^\text{err}$, while $C(\bm{\theta})$ is relevant to errors outside. 

\begin{figure}[t]
\includegraphics[width=0.6\columnwidth]{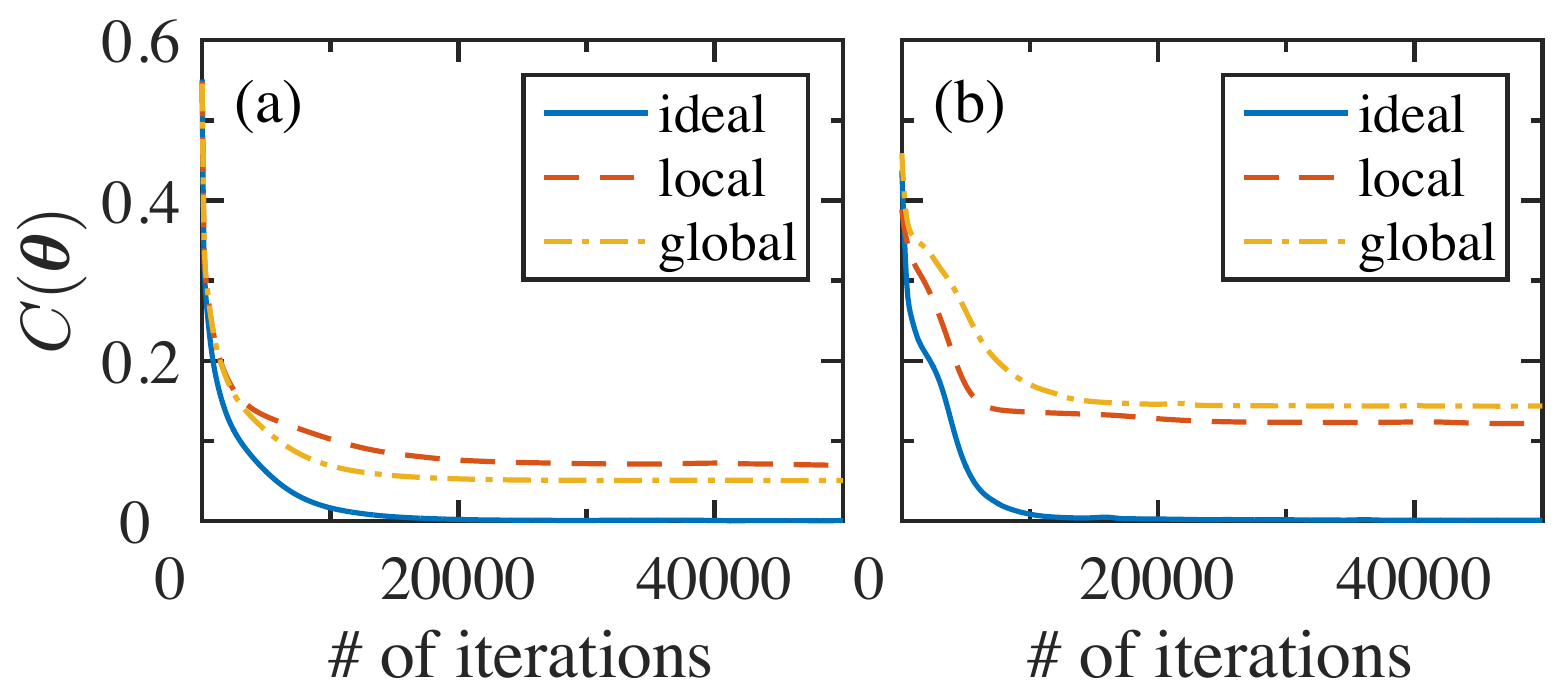}
\caption{Training curves of the cost function for projection to the junk subspace $C(\bm{\theta})$ for three different situations as indicated. Panel (a) shows results for the first stage of training, while panel (b) the second stage. Solid lines: ideal states with $\varepsilon=0$. Dash-dotted lines: global depolarization noise with $\varepsilon=0.1$. Dashed lines: local depolarization noise with $\varepsilon=0.1$.}
\label{fig:cv}
\end{figure}

\begin{figure} [h]
\includegraphics[width=0.56\columnwidth]{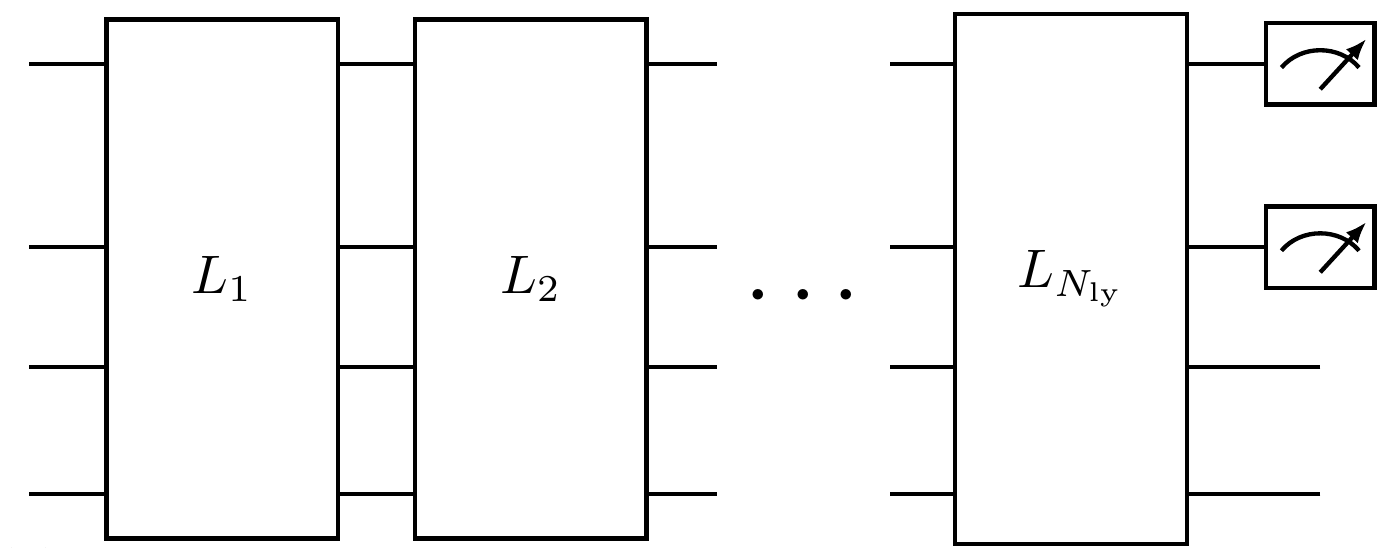}
\caption{Circuit structure for the single-stage training method.}
\label{fig:enc_sup}
\end{figure}

\begin{figure} [t]
\includegraphics[width=0.56\columnwidth]{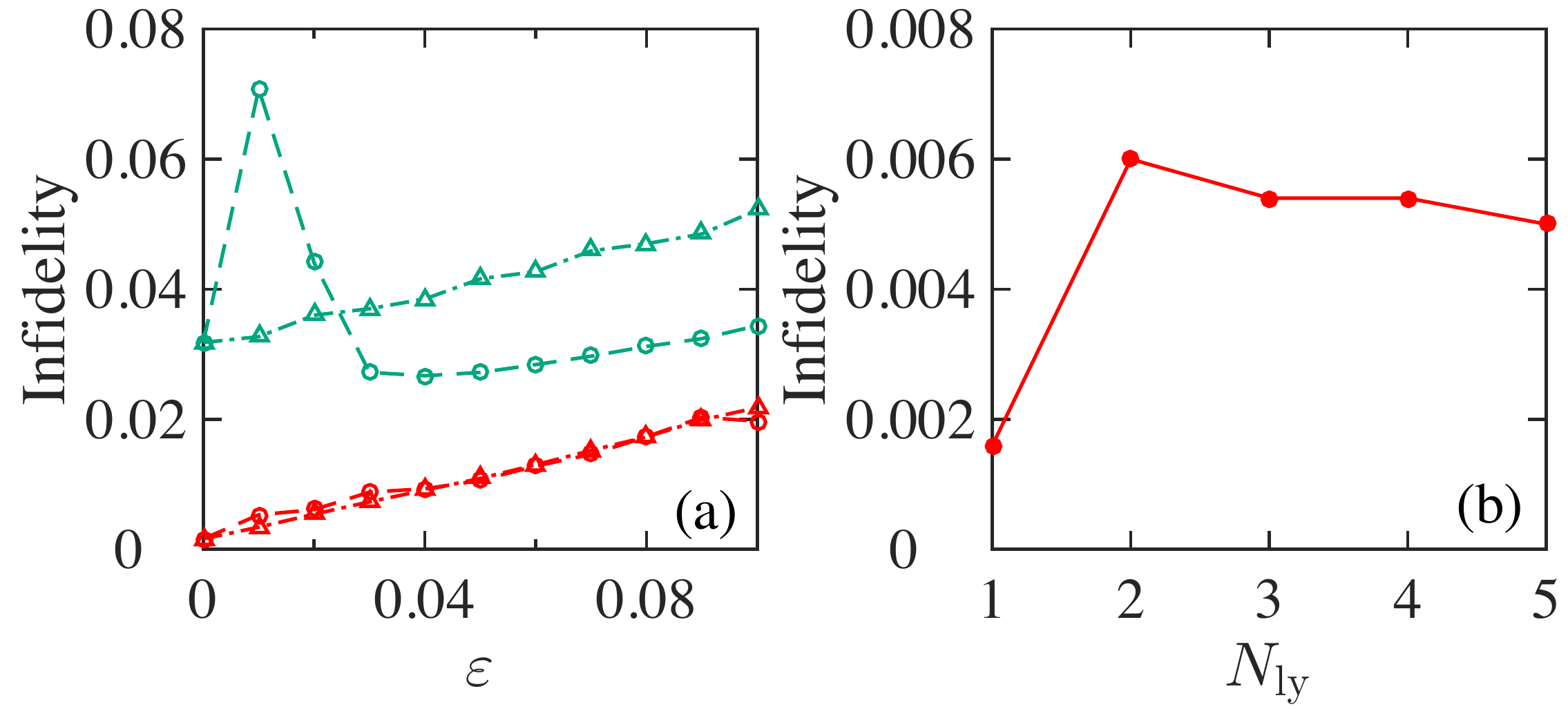}
\caption{ (a) Comparison of the performances of the two-stage method and the single-stage method. Green lines represent the single-stage method; red lines represent the two-stage circuit and training process used in the main text. Dashed lines and circles represent results with the local depolarization noise, dash-dotted lines and triangles represent those with global depolarization noise. For both methods, we set $N_{\text{ly}}=1$. (b) Performance versus number of layers $N_{\text{ly}}$. We use the two-stage method and assume no noise is present ($\varepsilon=0$).}
\label{fig:dr_ly}
\end{figure}

\subsection{Training stages and number of layers}
In the main text, we train a detection-based quantum autoencoder in two stages (``two-stage method''). In this subsection, we address two questions: (1) why the two-stage method is superior than that with a single stage in training; and (2) why $N_{\text{ly}}=1$ is optimal.

The circuit structure of a detection-based quantum autoencoder trained in a single stage (``single-stage method'') is shown in Fig.~\ref{fig:enc_sup}, with the structure of each layer identical to the two-stage method [cf. Fig.~\ref{fig:c}(b) in the main text]. Similar to the two-stage method, the circuit is trained to compress the input data to the latent subspace spanned by $\{|0\,0\,a_1a_2\rangle\}$ with $a_i=0,1$. In Fig.~\ref{fig:dr_ly}(a), the performances of the single- and two-stage methods are compared. Obviously, the two-stage method has a superior performance as it consistently produces lower infidelities without fluctuation over the parameter $\varepsilon$, while the single-stage method under local depolarization noise spikes at $\varepsilon=0.01$, indicating instabilities in the training process. This is because the smaller the subspace is, the more difficult one can compress the original states into it. Dividing the training process into multiple stages can help to avoid the optimization getting trapped in a poor local minimum.

We now restrict to the two-stage method and consider the effect of different number of layers. As shown in Fig.~\ref{fig:dr_ly}(b), the circuit with $N_{\text{ly}}=1$ has the lowest infidelity, and there is no improvement  as $N_{\text{ly}}$ increases. We conjecture that this is because the circuit with $N_{\text{ly}}=1$ is already sufficient to perform the compression, but as the number of parameters increase, the cost function is more likely to get trapped in a poor local minima during the training process. Nevertheless, we believe that for more complicated quantum data with higher dimensionality, more layers are necessary and one must carefully select the optimal $N_{\text{ly}}$ in the detection-based quantum autoencoder.

\begin{figure}[t]
\includegraphics[width=0.56\columnwidth]{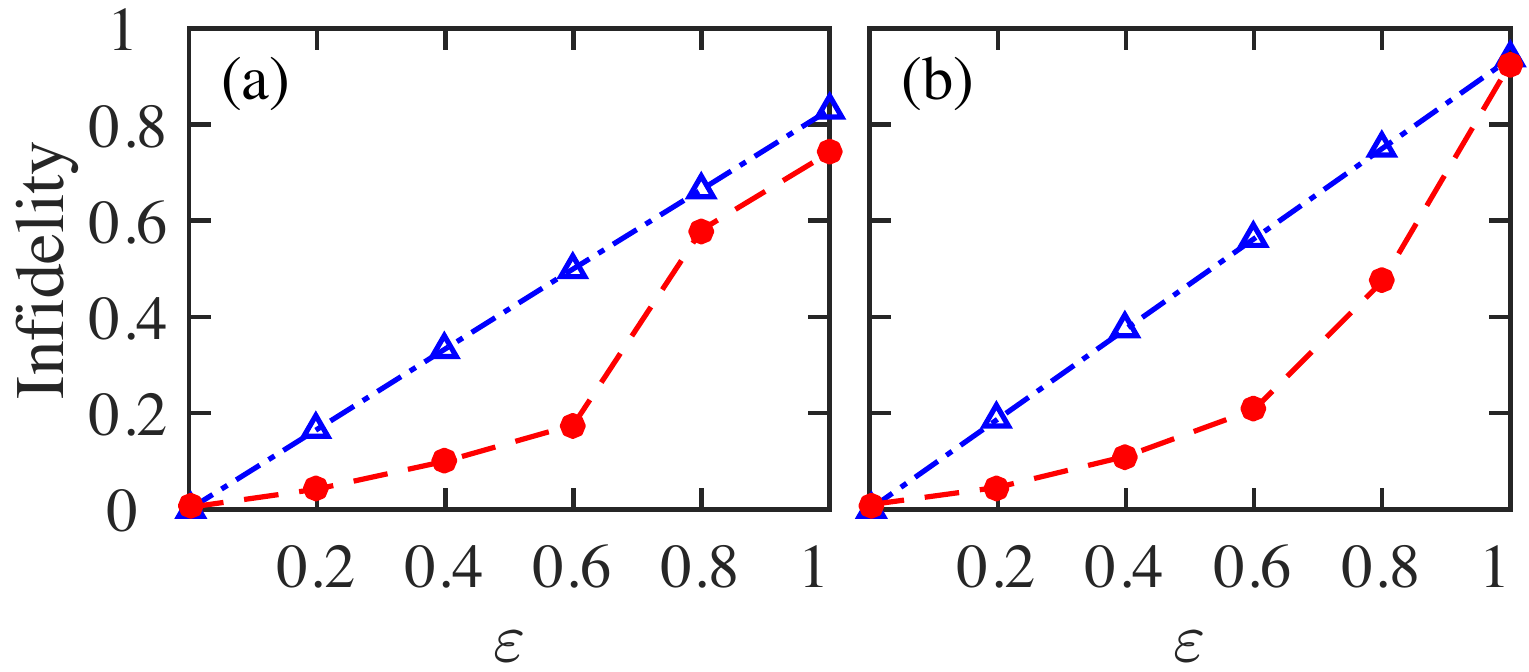}
\caption{Performance of a detection-based quantum autoencoder for a broad range of noise with (a) local depolarization noise $\widetilde{\mathcal{E}}_{\text{lc}}(\cdot)$ and (b) global depolarization noise $\widetilde{\mathcal{E}}_{\text{gl}}(\cdot)$. We set $n=4$. Blue lines with triangles represent the uncorrected states and red lines with circles represent the corrected state with one layer of autoencoder.}
\label{eq:large}
\end{figure}

\begin{figure}[t]
\includegraphics[width=0.56\columnwidth]{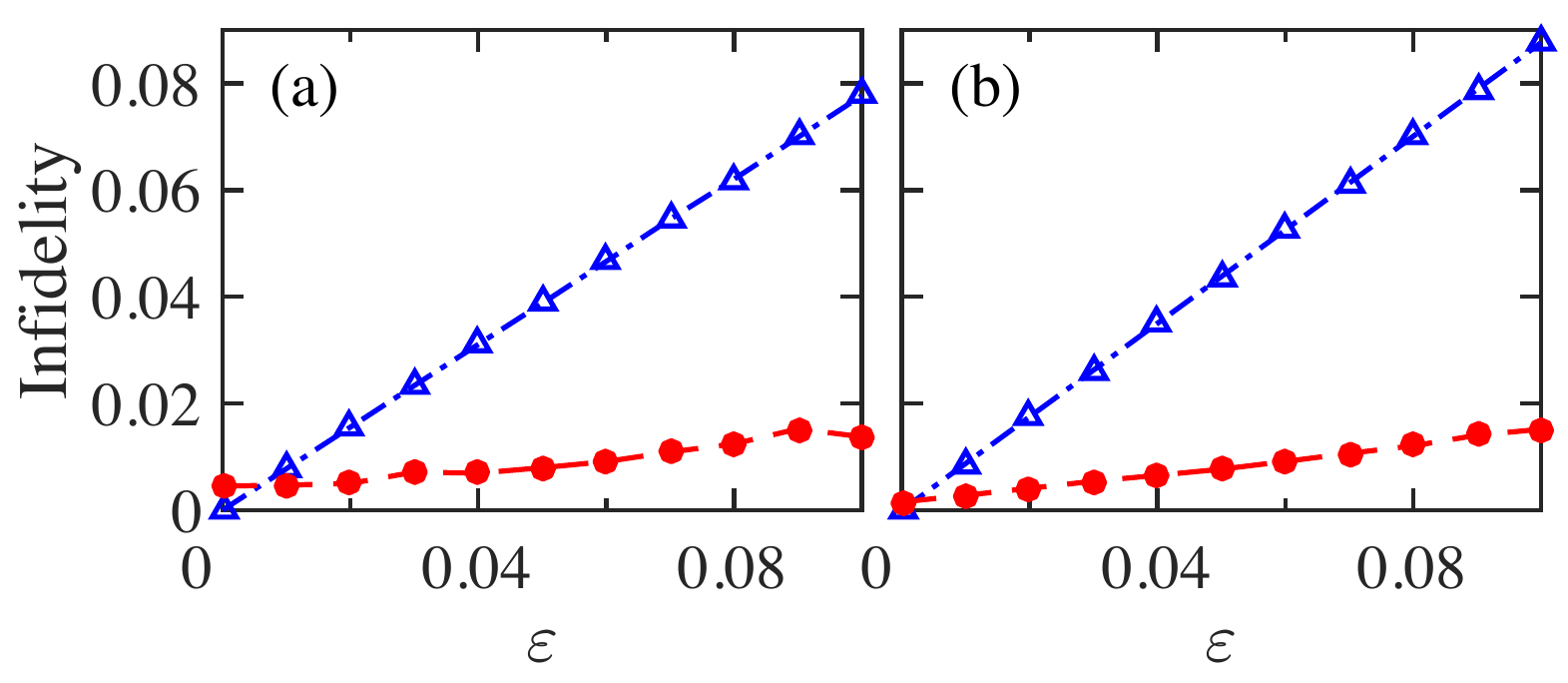}
\caption{Performance of a detection-based quantum autoencoder for mixed states with (a) local depolarization noise $\widetilde{\mathcal{E}}_{\text{lc}}(\cdot)$ and (d) global depolarization noise $\widetilde{\mathcal{E}}_{\text{gl}}(\cdot)$. We set $n=4$. Blue lines with triangles represent the uncorrected states and red lines with circles represent the corrected state with one layer of autoencoder.}
\label{fig:mixed}
\end{figure}

\begin{figure}[t]
\includegraphics[width=0.56\columnwidth]{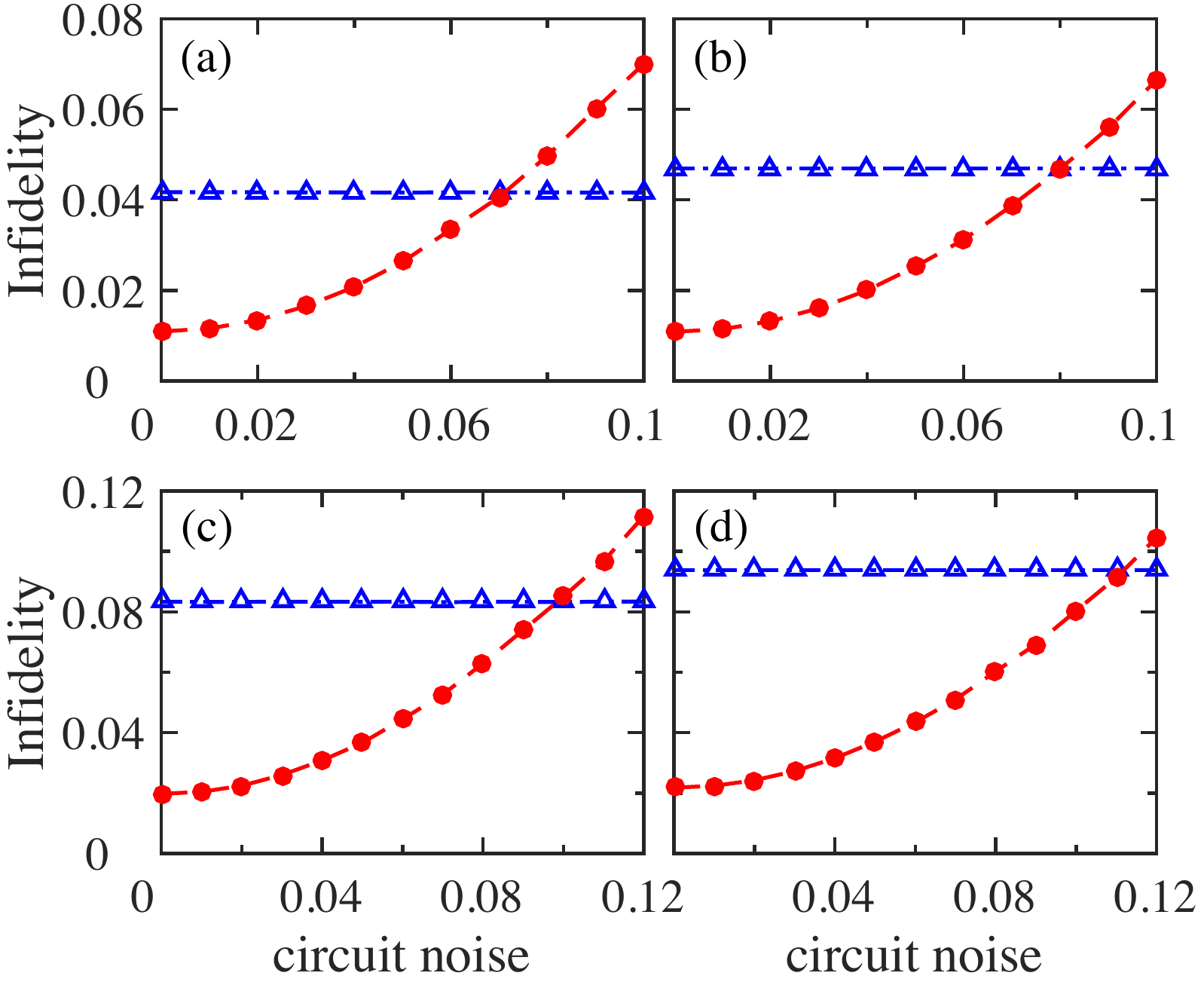}
\caption{Performance of detection-based quantum autoencoders with noisy quantum circuits. (a) local depolarization noise with $\varepsilon=0.05$; (b) global depolarization noise with $\varepsilon=0.05$; (c) local depolarization noise with $\varepsilon=0.1$; (d) global depolarization noise with $\varepsilon=0.1$. Blue lines with triangles represent the uncorrected states and red lines with circles represent the corrected state with one layer of autoencoder.}\label{fig:circuit_noise}
\end{figure}

\subsection{Large noise cases}
In this section, we consider the performance of our autoencoder under higher noises levels up to $\varepsilon=1$. The noises are applied to both training and testing phases.  As shown in Supplementary Fig.~\ref{eq:large}, the error-mitigated states shows significant improvement when $\varepsilon\leqslant0.6$ for both local and global depolarization noise, and
$\varepsilon\leqslant0.8$ for global depolarization noise. Therefore, our detection-based quantum autoencoder also works for reasonably large noises.

\subsection{Mixed states}
In Supplementary Fig.~\ref{fig:mixed}, we show results for error-mitigation on mixed states. We take a quantum autoencoder well-trained with pure-state data under the procedure explained in the main text,  then directly apply it to mixed-state data. The ideal states are $\rho_{\text{mix}}=p_1|\psi_1\rangle\langle\psi_1|+(1-p_1)|\psi_2\rangle\langle\psi_2|$, with two randomly generated pure $W$-class states $|\psi_{1,2}\rangle$ and probability $p_1\in[0,1]$ drawn from uniform distribution.  As can be seen, application of the quantum autoencoder reduces the infidelity substantially. Therefore our method should work for mixed states.

\subsection{Noisy quantum circuit}
In practice, the quantum circuit of the quantum autoencoder may not be ideal. It is therefore important to see whether the quantum autoencoder would still be able to improve fidelities with noisy quantum circuits. As an example, we consider a case where parameters in the circuits are deteriorated by noises. For each parameter of the ansatz $\theta_i$, we introduce a Gaussian noise  $\theta_i\rightarrow\theta_i+\delta\theta_i$, with $\delta\theta_i\sim\mathcal{N}(0,\varepsilon^2_{\text{circuit}})$. As can be seen in Supplementary Fig.~\ref{fig:circuit_noise}, the quantum autoencoder provide obvious improvement when the circuit noise level $\varepsilon_{\text{circuit}}$ is not high, roughly $\varepsilon_{\text{circuit}}\lesssim\varepsilon$.

\section{Results for hydrogen molecule}\label{app:cm1}

In this section we provide another example applying the detection-based quantum autoencoder.
In simulation of quantum chemistry, one typically needs to encode a molecular system, which has $N$ electrons occupying $M$ orbitals (with $M>N$), to a simulating qubit system. For example, under the Jordan-Wigner mapping, each qubit of the simulating system represents an orbital and the state $|0\rangle$ ($|1\rangle$) represents that the orbital is unoccupied (occupied). One of the simplest molecular system is the hydrogen molecule (H$_2$) whose quantum state can be described by \cite{SMcArdle.19},
\begin{equation}
\left|\psi_{\text{H}_2}\right\rangle=\alpha_1|0101\rangle+\alpha_2|1010\rangle+\alpha_3|1001\rangle+\alpha_4|0110\rangle,\label{eq:H}
\end{equation}
with certain values of $\alpha_i$.

We assume that $\left|\psi_{\text{H}_2}\right\rangle$ can be prepared for arbitrary $\alpha_i$ subject to noise. With the same programmable circuit structure and training process as in the main text for $W$ class states, we obtain a well-trained quantum autoencoder for Eq.~\eqref{eq:H}. The error-mitigation effect is shown in  Fig.~\ref{fig:H}. For both global and local depolarization noise models, when $\varepsilon>0$, the corrected  states have much lower infidelities compared to the uncorrected data, demonstrating the power of our detection-based method.

\begin{figure}[t]
\includegraphics[width=0.56\columnwidth]{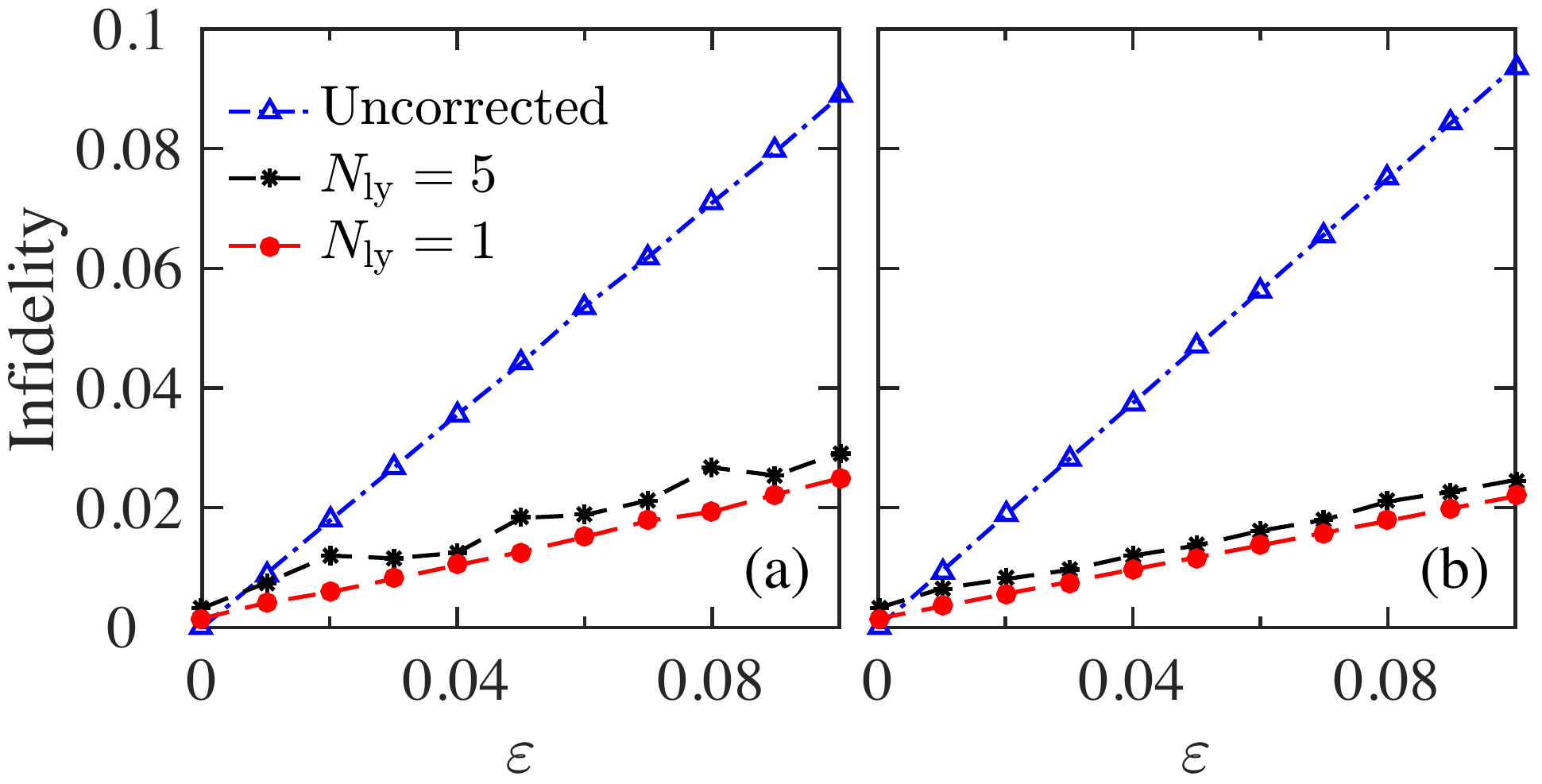}
\caption{ Performance of well-trained quantum autoencoders for quantum state described in Eq.~\eqref{eq:H} with (a) local and (d) global depolarization noise.}
\label{fig:H}
\end{figure}
\begin{figure}[t]
\includegraphics[width=0.56\columnwidth]{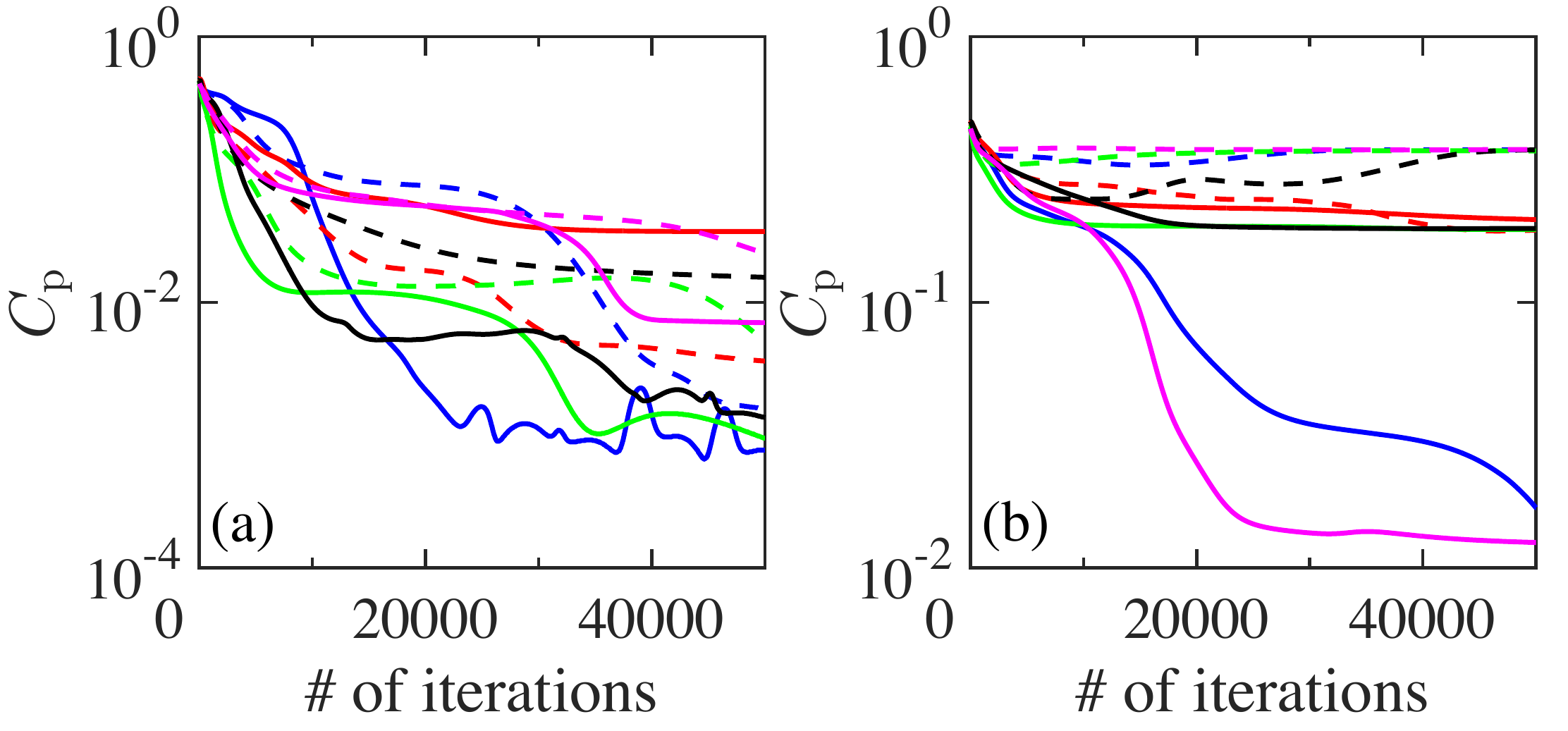}
\caption{Comparison of the training process with cost function Eq.~\eqref{eq:purity} and Eq.~\eqref{eq:cost}, which are represented with solid lines and dashed lines respectively. The $y$-axis represents the cost function $C_{\text{p}}(\bm{\theta})$. (a) The first stage of the two-stage training method. (b) The single-stage training method. Different colors represent different initial guess of $\bm{\theta}$.  }
\label{fig:purity}
\end{figure}

\section{optimization based on purity}
As mentioned in the main text, instead of determining the form of $\mathcal{L}$ \emph{a priori}, one may just fix the dimension of $\mathcal{L}$ and allow its form  to change with training. In our original protocol, the circuit is trained to project all $(n-m)$ qubits to state $|0\rangle$, representing the projection to the latent subspace. In fact, up to a local unitary, this is equivalent to projecting them to any disentangled pure states. So we may change the cost function [Eq.~\eqref{eq:cost}] to  a form relating to the purity and entanglement of the $(n-m)$-qubit system. Take the two-stage method for $n=4$ W class states as an example, in the first stage, the cost function can be changed to the purity of the first qubit:
\begin{equation}
C_{\text{p}}(\bm{\theta})=1-\text{Tr}\left[\text{Tr}_{2,3,4}\left[\sigma(\bm{\theta})\right]^2\right],\label{eq:purity}
\end{equation}
where $\text{Tr}_{2,3,4}$ represents the partial trace over the $2$nd, $3$rd and $4$th qubits. In case $C_{\text{p}}(\bm{\theta})$ is minimized to zero, the first qubit can be transferred to state $|0\rangle$ with a single qubit unitary. Note that the training process in the second stage can be similar. For the one-stage method, the four-qubit states are compressed to two-qubit states directly. To maximize the purity and minimize the entanglement of the discarded qubits, we define the cost function as one minus the average purity of the first and second qubit 
\begin{equation}
C_{\text{p}}(\bm{\theta})=1-\frac{1}{2}\text{Tr}\left[\text{Tr}_{1,3,4}\left[\sigma(\bm{\theta})\right]^2\right]-\frac{1}{2}\text{Tr}\left[\text{Tr}_{2,3,4}\left[\sigma(\bm{\theta})\right]^2\right].\label{eq:purity}
\end{equation}
In Fig.~\ref{fig:purity}, we compare the optimization process with the cost function $C(\bm{\theta})$ defined in Eq.~\eqref{eq:cost} and $C_{\text{p}}(\bm{\theta})$ with the same initial guess of $\bm{\theta}$. 
 In most cases, the optimization with $C_{\text{p}}(\bm{\theta})$ converges faster and ends up with higher final purity. Therefore, we believe that this method can help improving the training performance in future studies. Moreover, the purity cost function is also potentially useful for neural-network-based quantum autoencoder~\cite{Bondarenko.19,Beer.20}, as it can avoid non-pure output states. 
 
However, the purity cost function also has its limitation. Firstly, the quantum data are assumed to be pure, and it is not straightforward to generalize the purity cost function to accommodate mixed states.  Moreover, measuring purity is more difficult than projecting the data to the latent subspace as implemented in the main text.


\begin{thebibliography}{47}%
\makeatletter
\providecommand \@ifxundefined [1]{%
 \@ifx{#1\undefined}
}%
\providecommand \@ifnum [1]{%
 \ifnum #1\expandafter \@firstoftwo
 \else \expandafter \@secondoftwo
 \fi
}%
\providecommand \@ifx [1]{%
 \ifx #1\expandafter \@firstoftwo
 \else \expandafter \@secondoftwo
 \fi
}%
\providecommand \natexlab [1]{#1}%
\providecommand \enquote  [1]{``#1''}%
\providecommand \bibnamefont  [1]{#1}%
\providecommand \bibfnamefont [1]{#1}%
\providecommand \citenamefont [1]{#1}%
\providecommand \href@noop [0]{\@secondoftwo}%
\providecommand \href [0]{\begingroup \@sanitize@url \@href}%
\providecommand \@href[1]{\@@startlink{#1}\@@href}%
\providecommand \@@href[1]{\endgroup#1\@@endlink}%
\providecommand \@sanitize@url [0]{\catcode `\\12\catcode `\$12\catcode
  `\&12\catcode `\#12\catcode `\^12\catcode `\_12\catcode `\%12\relax}%
\providecommand \@@startlink[1]{}%
\providecommand \@@endlink[0]{}%
\providecommand \url  [0]{\begingroup\@sanitize@url \@url }%
\providecommand \@url [1]{\endgroup\@href {#1}{\urlprefix }}%
\providecommand \urlprefix  [0]{URL }%
\providecommand \Eprint [0]{\href }%
\providecommand \doibase [0]{http://dx.doi.org/}%
\providecommand \selectlanguage [0]{\@gobble}%
\providecommand \bibinfo  [0]{\@secondoftwo}%
\providecommand \bibfield  [0]{\@secondoftwo}%
\providecommand \translation [1]{[#1]}%
\providecommand \BibitemOpen [0]{}%
\providecommand \bibitemStop [0]{}%
\providecommand \bibitemNoStop [0]{.\EOS\space}%
\providecommand \EOS [0]{\spacefactor3000\relax}%
\providecommand \BibitemShut  [1]{\csname bibitem#1\endcsname}%
\let\auto@bib@innerbib\@empty
\bibitem [{\citenamefont {Nielsen}\ and\ \citenamefont
  {Chuang}(2000)}]{Nielsen.02}%
  \BibitemOpen
  \bibfield  {author} {\bibinfo {author} {\bibfnamefont {M.~A.}\ \bibnamefont
  {Nielsen}}\ and\ \bibinfo {author} {\bibfnamefont {I.}~\bibnamefont
  {Chuang}},\ }\href@noop {} {\emph {\bibinfo {title} {Quantum Computation and
  Quantum information}}}\ (\bibinfo  {publisher} {Cambridge University Press,
  Cambridge},\ \bibinfo {year} {2000})\BibitemShut {NoStop}%
\bibitem [{\citenamefont {Altepeter}\ \emph {et~al.}(2004)\citenamefont
  {Altepeter}, \citenamefont {Hadley}, \citenamefont {Wendelken}, \citenamefont
  {Berglund},\ and\ \citenamefont {Kwiat}}]{Altepeter.04}%
  \BibitemOpen
  \bibfield  {author} {\bibinfo {author} {\bibfnamefont {J.}~\bibnamefont
  {Altepeter}}, \bibinfo {author} {\bibfnamefont {P.}~\bibnamefont {Hadley}},
  \bibinfo {author} {\bibfnamefont {S.}~\bibnamefont {Wendelken}}, \bibinfo
  {author} {\bibfnamefont {A.}~\bibnamefont {Berglund}}, \ and\ \bibinfo
  {author} {\bibfnamefont {P.}~\bibnamefont {Kwiat}},\ }\href
  {https://journals.aps.org/prl/abstract/10.1103/PhysRevLett.92.147901}
  {\bibfield  {journal} {\bibinfo  {journal} {Phys. Rev. Lett.}\ }\textbf
  {\bibinfo {volume} {92}},\ \bibinfo {pages} {147901} (\bibinfo {year}
  {2004})}\BibitemShut {NoStop}%
\bibitem [{\citenamefont {Xue}\ and\ \citenamefont {Xiao}(2006)}]{Xue.06}%
  \BibitemOpen
  \bibfield  {author} {\bibinfo {author} {\bibfnamefont {P.}~\bibnamefont
  {Xue}}\ and\ \bibinfo {author} {\bibfnamefont {Y.-F.}\ \bibnamefont {Xiao}},\
  }\href {https://journals.aps.org/prl/abstract/10.1103/PhysRevLett.97.140501}
  {\bibfield  {journal} {\bibinfo  {journal} {Phys. Rev. Lett.}\ }\textbf
  {\bibinfo {volume} {97}},\ \bibinfo {pages} {140501} (\bibinfo {year}
  {2006})}\BibitemShut {NoStop}%
\bibitem [{\citenamefont {Fong}\ and\ \citenamefont
  {Wandzura}(2011)}]{Fong.11}%
  \BibitemOpen
  \bibfield  {author} {\bibinfo {author} {\bibfnamefont {B.~H.}\ \bibnamefont
  {Fong}}\ and\ \bibinfo {author} {\bibfnamefont {S.~M.}\ \bibnamefont
  {Wandzura}},\ }\href {https://arxiv.org/abs/1102.2909} {\bibfield  {journal}
  {\bibinfo  {journal} {Quantum Inf. Comput.}\ }\textbf {\bibinfo {volume}
  {11}},\ \bibinfo {pages} {1003} (\bibinfo {year} {2011})}\BibitemShut
  {NoStop}%
\bibitem [{\citenamefont {Friesen}\ \emph {et~al.}(2017)\citenamefont
  {Friesen}, \citenamefont {Ghosh}, \citenamefont {Eriksson},\ and\
  \citenamefont {Coppersmith}}]{Friesen.17}%
  \BibitemOpen
  \bibfield  {author} {\bibinfo {author} {\bibfnamefont {M.}~\bibnamefont
  {Friesen}}, \bibinfo {author} {\bibfnamefont {J.}~\bibnamefont {Ghosh}},
  \bibinfo {author} {\bibfnamefont {M.~A.}\ \bibnamefont {Eriksson}}, \ and\
  \bibinfo {author} {\bibfnamefont {S.~N.}\ \bibnamefont {Coppersmith}},\
  }\href {https://www.nature.com/articles/ncomms15923} {\bibfield  {journal}
  {\bibinfo  {journal} {Nat. Commun.}\ }\textbf {\bibinfo {volume} {8}},\
  \bibinfo {pages} {15923} (\bibinfo {year} {2017})}\BibitemShut {NoStop}%
\bibitem [{\citenamefont {Preskill}(2018)}]{Preskill.18}%
  \BibitemOpen
  \bibfield  {author} {\bibinfo {author} {\bibfnamefont {J.}~\bibnamefont
  {Preskill}},\ }\href
  {https://quantum-journal.org/papers/q-2018-08-06-79/?fbclid=IwAR0FRb9N2fas7ETWu2M40OS6prXB5QvFME_WRELpm2CAUcccVIEzA_UmLn4}
  {\bibfield  {journal} {\bibinfo  {journal} {Quantum}\ }\textbf {\bibinfo
  {volume} {2}},\ \bibinfo {pages} {79} (\bibinfo {year} {2018})}\BibitemShut
  {NoStop}%
\bibitem [{\citenamefont {Peruzzo}\ \emph {et~al.}(2014)\citenamefont
  {Peruzzo}, \citenamefont {McClean}, \citenamefont {Shadbolt}, \citenamefont
  {Yung}, \citenamefont {Zhou}, \citenamefont {Love}, \citenamefont
  {Aspuru-Guzik},\ and\ \citenamefont {O’brien}}]{Peruzzo.14}%
  \BibitemOpen
  \bibfield  {author} {\bibinfo {author} {\bibfnamefont {A.}~\bibnamefont
  {Peruzzo}}, \bibinfo {author} {\bibfnamefont {J.}~\bibnamefont {McClean}},
  \bibinfo {author} {\bibfnamefont {P.}~\bibnamefont {Shadbolt}}, \bibinfo
  {author} {\bibfnamefont {M.-H.}\ \bibnamefont {Yung}}, \bibinfo {author}
  {\bibfnamefont {X.-Q.}\ \bibnamefont {Zhou}}, \bibinfo {author}
  {\bibfnamefont {P.~J.}\ \bibnamefont {Love}}, \bibinfo {author}
  {\bibfnamefont {A.}~\bibnamefont {Aspuru-Guzik}}, \ and\ \bibinfo {author}
  {\bibfnamefont {J.~L.}\ \bibnamefont {O’brien}},\ }\href
  {https://doi.org/10.1038/ncomms5213} {\bibfield  {journal} {\bibinfo
  {journal} {Nat. Comm.}\ }\textbf {\bibinfo {volume} {5}},\ \bibinfo {pages}
  {4213} (\bibinfo {year} {2014})}\BibitemShut {NoStop}%
\bibitem [{\citenamefont {Kandala}\ \emph {et~al.}(2017)\citenamefont
  {Kandala}, \citenamefont {Mezzacapo}, \citenamefont {Temme}, \citenamefont
  {Takita}, \citenamefont {Brink}, \citenamefont {Chow},\ and\ \citenamefont
  {Gambetta}}]{Kandala.17}%
  \BibitemOpen
  \bibfield  {author} {\bibinfo {author} {\bibfnamefont {A.}~\bibnamefont
  {Kandala}}, \bibinfo {author} {\bibfnamefont {A.}~\bibnamefont {Mezzacapo}},
  \bibinfo {author} {\bibfnamefont {K.}~\bibnamefont {Temme}}, \bibinfo
  {author} {\bibfnamefont {M.}~\bibnamefont {Takita}}, \bibinfo {author}
  {\bibfnamefont {M.}~\bibnamefont {Brink}}, \bibinfo {author} {\bibfnamefont
  {J.~M.}\ \bibnamefont {Chow}}, \ and\ \bibinfo {author} {\bibfnamefont
  {J.~M.}\ \bibnamefont {Gambetta}},\ }\href
  {https://www.nature.com/articles/nature23879?sf114016447=1} {\bibfield
  {journal} {\bibinfo  {journal} {Nature (London)}\ }\textbf {\bibinfo {volume}
  {549}},\ \bibinfo {pages} {242} (\bibinfo {year} {2017})}\BibitemShut
  {NoStop}%
\bibitem [{\citenamefont {Farhi}\ \emph {et~al.}(2014)\citenamefont {Farhi},
  \citenamefont {Goldstone},\ and\ \citenamefont {Gutmann}}]{Farhi.14}%
  \BibitemOpen
  \bibfield  {author} {\bibinfo {author} {\bibfnamefont {E.}~\bibnamefont
  {Farhi}}, \bibinfo {author} {\bibfnamefont {J.}~\bibnamefont {Goldstone}}, \
  and\ \bibinfo {author} {\bibfnamefont {S.}~\bibnamefont {Gutmann}},\ }\href
  {https://arxiv.org/abs/1411.4028} {\bibfield  {journal} {\bibinfo  {journal}
  {arXiv:1411.4028}\ } (\bibinfo {year} {2014})}\BibitemShut {NoStop}%
\bibitem [{\citenamefont {Farhi}\ and\ \citenamefont
  {Harrow}(2016)}]{Farhi.16}%
  \BibitemOpen
  \bibfield  {author} {\bibinfo {author} {\bibfnamefont {E.}~\bibnamefont
  {Farhi}}\ and\ \bibinfo {author} {\bibfnamefont {A.~W.}\ \bibnamefont
  {Harrow}},\ }\href {https://arxiv.org/abs/1602.07674} {\bibfield  {journal}
  {\bibinfo  {journal} {arXiv:1602.07674}\ } (\bibinfo {year}
  {2016})}\BibitemShut {NoStop}%
\bibitem [{\citenamefont {McArdle}\ \emph {et~al.}(2020)\citenamefont
  {McArdle}, \citenamefont {Endo}, \citenamefont {Aspuru-Guzik}, \citenamefont
  {Benjamin},\ and\ \citenamefont {Yuan}}]{McArdle.20}%
  \BibitemOpen
  \bibfield  {author} {\bibinfo {author} {\bibfnamefont {S.}~\bibnamefont
  {McArdle}}, \bibinfo {author} {\bibfnamefont {S.}~\bibnamefont {Endo}},
  \bibinfo {author} {\bibfnamefont {A.}~\bibnamefont {Aspuru-Guzik}}, \bibinfo
  {author} {\bibfnamefont {S.~C.}\ \bibnamefont {Benjamin}}, \ and\ \bibinfo
  {author} {\bibfnamefont {X.}~\bibnamefont {Yuan}},\ }\href
  {https://link.aps.org/doi/10.1103/RevModPhys.92.015003} {\bibfield  {journal}
  {\bibinfo  {journal} {Rev. Mod. Phys.}\ }\textbf {\bibinfo {volume} {92}},\
  \bibinfo {pages} {015003} (\bibinfo {year} {2020})}\BibitemShut {NoStop}%
\bibitem [{\citenamefont {Wold}\ \emph {et~al.}(1987)\citenamefont {Wold},
  \citenamefont {Esbensen},\ and\ \citenamefont {Geladi}}]{Wold.87}%
  \BibitemOpen
  \bibfield  {author} {\bibinfo {author} {\bibfnamefont {S.}~\bibnamefont
  {Wold}}, \bibinfo {author} {\bibfnamefont {K.}~\bibnamefont {Esbensen}}, \
  and\ \bibinfo {author} {\bibfnamefont {P.}~\bibnamefont {Geladi}},\ }\href
  {https://www.sciencedirect.com/science/article/abs/pii/0169743987800849}
  {\bibfield  {journal} {\bibinfo  {journal} {Chemom. Intell. Lab. Syst.}\
  }\textbf {\bibinfo {volume} {2}},\ \bibinfo {pages} {37} (\bibinfo {year}
  {1987})}\BibitemShut {NoStop}%
\bibitem [{\citenamefont {Vincent}\ \emph {et~al.}(2008)\citenamefont
  {Vincent}, \citenamefont {Larochelle}, \citenamefont {Bengio},\ and\
  \citenamefont {Manzagol}}]{Vincent.08}%
  \BibitemOpen
  \bibfield  {author} {\bibinfo {author} {\bibfnamefont {P.}~\bibnamefont
  {Vincent}}, \bibinfo {author} {\bibfnamefont {H.}~\bibnamefont {Larochelle}},
  \bibinfo {author} {\bibfnamefont {Y.}~\bibnamefont {Bengio}}, \ and\ \bibinfo
  {author} {\bibfnamefont {P.-A.}\ \bibnamefont {Manzagol}},\ }in\ \href
  {https://dl.acm.org/doi/abs/10.1145/1390156.1390294?casa_token=uHHP33Q8lLsAAAAA%3AGRiJiiQ4pMhGkHWQ4iLQZ2u0mqws2gFFZGmO1oD2j-MfuPq56BLA3NOOIDjP1TeMxwncOrgAA54Kew}
  {\emph {\bibinfo {booktitle} {ICML}}}\ (\bibinfo {year} {2008})\ pp.\
  \bibinfo {pages} {1096--1103}\BibitemShut {NoStop}%
\bibitem [{\citenamefont {Vincent}\ \emph {et~al.}(2010)\citenamefont
  {Vincent}, \citenamefont {Larochelle}, \citenamefont {Lajoie}, \citenamefont
  {Bengio},\ and\ \citenamefont {Manzagol}}]{Vincent.10}%
  \BibitemOpen
  \bibfield  {author} {\bibinfo {author} {\bibfnamefont {P.}~\bibnamefont
  {Vincent}}, \bibinfo {author} {\bibfnamefont {H.}~\bibnamefont {Larochelle}},
  \bibinfo {author} {\bibfnamefont {I.}~\bibnamefont {Lajoie}}, \bibinfo
  {author} {\bibfnamefont {Y.}~\bibnamefont {Bengio}}, \ and\ \bibinfo {author}
  {\bibfnamefont {P.-A.}\ \bibnamefont {Manzagol}},\ }\href
  {http://www.jmlr.org/papers/v11/vincent10a.html} {\bibfield  {journal}
  {\bibinfo  {journal} {J. Mach. Learn. Res.}\ }\textbf {\bibinfo {volume}
  {11}},\ \bibinfo {pages} {3371} (\bibinfo {year} {2010})}\BibitemShut
  {NoStop}%
\bibitem [{\citenamefont {Romero}\ \emph {et~al.}(2017)\citenamefont {Romero},
  \citenamefont {Olson},\ and\ \citenamefont {Aspuru-Guzik}}]{Romero.17}%
  \BibitemOpen
  \bibfield  {author} {\bibinfo {author} {\bibfnamefont {J.}~\bibnamefont
  {Romero}}, \bibinfo {author} {\bibfnamefont {J.~P.}\ \bibnamefont {Olson}}, \
  and\ \bibinfo {author} {\bibfnamefont {A.}~\bibnamefont {Aspuru-Guzik}},\
  }\href {https://iopscience.iop.org/article/10.1088/2058-9565/aa8072}
  {\bibfield  {journal} {\bibinfo  {journal} {Quantum. Sci. Technol.}\ }\textbf
  {\bibinfo {volume} {2}},\ \bibinfo {pages} {045001} (\bibinfo {year}
  {2017})}\BibitemShut {NoStop}%
\bibitem [{\citenamefont {Wan}\ \emph {et~al.}(2017)\citenamefont {Wan},
  \citenamefont {Dahlsten}, \citenamefont {Kristj{\'a}nsson}, \citenamefont
  {Gardner},\ and\ \citenamefont {Kim}}]{Wan.17}%
  \BibitemOpen
  \bibfield  {author} {\bibinfo {author} {\bibfnamefont {K.~H.}\ \bibnamefont
  {Wan}}, \bibinfo {author} {\bibfnamefont {O.}~\bibnamefont {Dahlsten}},
  \bibinfo {author} {\bibfnamefont {H.}~\bibnamefont {Kristj{\'a}nsson}},
  \bibinfo {author} {\bibfnamefont {R.}~\bibnamefont {Gardner}}, \ and\
  \bibinfo {author} {\bibfnamefont {M.}~\bibnamefont {Kim}},\ }\href
  {https://www.nature.com/articles/s41534-017-0032-4} {\bibfield  {journal}
  {\bibinfo  {journal} {npj Quantum Inf.}\ }\textbf {\bibinfo {volume} {3}},\
  \bibinfo {pages} {36} (\bibinfo {year} {2017})}\BibitemShut {NoStop}%
\bibitem [{\citenamefont {Zhao}\ and\ \citenamefont {Gao}(2019)}]{Zhao.19}%
  \BibitemOpen
  \bibfield  {author} {\bibinfo {author} {\bibfnamefont {C.}~\bibnamefont
  {Zhao}}\ and\ \bibinfo {author} {\bibfnamefont {X.-S.}\ \bibnamefont {Gao}},\
  }\href {https://arxiv.org/abs/2005.07601} {\bibfield  {journal} {\bibinfo
  {journal} {arXiv:1912.12660}\ } (\bibinfo {year} {2019})}\BibitemShut
  {NoStop}%
\bibitem [{\citenamefont {Beer}\ \emph {et~al.}(2020)\citenamefont {Beer},
  \citenamefont {Bondarenko}, \citenamefont {Farrelly}, \citenamefont
  {Osborne}, \citenamefont {Salzmann}, \citenamefont {Scheiermann},\ and\
  \citenamefont {Wolf}}]{Beer.20}%
  \BibitemOpen
  \bibfield  {author} {\bibinfo {author} {\bibfnamefont {K.}~\bibnamefont
  {Beer}}, \bibinfo {author} {\bibfnamefont {D.}~\bibnamefont {Bondarenko}},
  \bibinfo {author} {\bibfnamefont {T.}~\bibnamefont {Farrelly}}, \bibinfo
  {author} {\bibfnamefont {T.~J.}\ \bibnamefont {Osborne}}, \bibinfo {author}
  {\bibfnamefont {R.}~\bibnamefont {Salzmann}}, \bibinfo {author}
  {\bibfnamefont {D.}~\bibnamefont {Scheiermann}}, \ and\ \bibinfo {author}
  {\bibfnamefont {R.}~\bibnamefont {Wolf}},\ }\href
  {https://www.nature.com/articles/s41467-020-14454-2} {\bibfield  {journal}
  {\bibinfo  {journal} {Nat. Commun.}\ }\textbf {\bibinfo {volume} {11}},\
  \bibinfo {pages} {808} (\bibinfo {year} {2020})}\BibitemShut {NoStop}%
\bibitem [{\citenamefont {Achache}\ \emph {et~al.}(2020)\citenamefont
  {Achache}, \citenamefont {Horesh},\ and\ \citenamefont
  {Smolin}}]{Achache.20}%
  \BibitemOpen
  \bibfield  {author} {\bibinfo {author} {\bibfnamefont {T.}~\bibnamefont
  {Achache}}, \bibinfo {author} {\bibfnamefont {L.}~\bibnamefont {Horesh}}, \
  and\ \bibinfo {author} {\bibfnamefont {J.}~\bibnamefont {Smolin}},\ }\href
  {https://arxiv.org/abs/2012.14714} {\bibfield  {journal} {\bibinfo  {journal}
  {arXiv:2012.14714}\ } (\bibinfo {year} {2020})}\BibitemShut {NoStop}%
\bibitem [{\citenamefont {Cao}\ and\ \citenamefont {Wang}(2020)}]{Cao.20}%
  \BibitemOpen
  \bibfield  {author} {\bibinfo {author} {\bibfnamefont {C.}~\bibnamefont
  {Cao}}\ and\ \bibinfo {author} {\bibfnamefont {X.}~\bibnamefont {Wang}},\
  }\href {https://arxiv.org/abs/2012.08331} {\bibfield  {journal} {\bibinfo
  {journal} {arXiv:2012.08331}\ } (\bibinfo {year} {2020})}\BibitemShut
  {NoStop}%
\bibitem [{\citenamefont {Pepper}\ \emph {et~al.}(2019)\citenamefont {Pepper},
  \citenamefont {Tischler},\ and\ \citenamefont {Pryde}}]{Pepper.19}%
  \BibitemOpen
  \bibfield  {author} {\bibinfo {author} {\bibfnamefont {A.}~\bibnamefont
  {Pepper}}, \bibinfo {author} {\bibfnamefont {N.}~\bibnamefont {Tischler}}, \
  and\ \bibinfo {author} {\bibfnamefont {G.~J.}\ \bibnamefont {Pryde}},\ }\href
  {https://journals.aps.org/prl/abstract/10.1103/PhysRevLett.122.060501}
  {\bibfield  {journal} {\bibinfo  {journal} {Phys. Rev. Lett.}\ }\textbf
  {\bibinfo {volume} {122}},\ \bibinfo {pages} {060501} (\bibinfo {year}
  {2019})}\BibitemShut {NoStop}%
\bibitem [{\citenamefont {Huang}\ \emph {et~al.}(2020)\citenamefont {Huang},
  \citenamefont {Ma}, \citenamefont {Yin}, \citenamefont {Tang}, \citenamefont
  {Dong}, \citenamefont {Chen}, \citenamefont {Xiang}, \citenamefont {Li},\
  and\ \citenamefont {Guo}}]{Huang.19}%
  \BibitemOpen
  \bibfield  {author} {\bibinfo {author} {\bibfnamefont {C.-J.}\ \bibnamefont
  {Huang}}, \bibinfo {author} {\bibfnamefont {H.}~\bibnamefont {Ma}}, \bibinfo
  {author} {\bibfnamefont {Q.}~\bibnamefont {Yin}}, \bibinfo {author}
  {\bibfnamefont {J.-F.}\ \bibnamefont {Tang}}, \bibinfo {author}
  {\bibfnamefont {D.}~\bibnamefont {Dong}}, \bibinfo {author} {\bibfnamefont
  {C.}~\bibnamefont {Chen}}, \bibinfo {author} {\bibfnamefont {G.-Y.}\
  \bibnamefont {Xiang}}, \bibinfo {author} {\bibfnamefont {C.-F.}\ \bibnamefont
  {Li}}, \ and\ \bibinfo {author} {\bibfnamefont {G.-C.}\ \bibnamefont {Guo}},\
  }\href {\doibase 10.1103/PhysRevA.102.032412} {\bibfield  {journal} {\bibinfo
   {journal} {Phys. Rev. A}\ }\textbf {\bibinfo {volume} {102}},\ \bibinfo
  {pages} {032412} (\bibinfo {year} {2020})}\BibitemShut {NoStop}%
\bibitem [{\citenamefont {Bondarenko}\ and\ \citenamefont
  {Feldmann}(2020)}]{Bondarenko.19}%
  \BibitemOpen
  \bibfield  {author} {\bibinfo {author} {\bibfnamefont {D.}~\bibnamefont
  {Bondarenko}}\ and\ \bibinfo {author} {\bibfnamefont {P.}~\bibnamefont
  {Feldmann}},\ }\href
  {https://journals.aps.org/prl/abstract/10.1103/PhysRevLett.124.130502}
  {\bibfield  {journal} {\bibinfo  {journal} {Phys. Rev. Lett.}\ }\textbf
  {\bibinfo {volume} {124}},\ \bibinfo {pages} {130502} (\bibinfo {year}
  {2020})}\BibitemShut {NoStop}%
\bibitem [{\citenamefont {Li}\ and\ \citenamefont {Benjamin}(2017)}]{Li.17}%
  \BibitemOpen
  \bibfield  {author} {\bibinfo {author} {\bibfnamefont {Y.}~\bibnamefont
  {Li}}\ and\ \bibinfo {author} {\bibfnamefont {S.~C.}\ \bibnamefont
  {Benjamin}},\ }\href {https://link.aps.org/doi/10.1103/PhysRevX.7.021050}
  {\bibfield  {journal} {\bibinfo  {journal} {Phys. Rev. X}\ }\textbf {\bibinfo
  {volume} {7}},\ \bibinfo {pages} {021050} (\bibinfo {year}
  {2017})}\BibitemShut {NoStop}%
\bibitem [{\citenamefont {Temme}\ \emph {et~al.}(2017)\citenamefont {Temme},
  \citenamefont {Bravyi},\ and\ \citenamefont {Gambetta}}]{Temme.17}%
  \BibitemOpen
  \bibfield  {author} {\bibinfo {author} {\bibfnamefont {K.}~\bibnamefont
  {Temme}}, \bibinfo {author} {\bibfnamefont {S.}~\bibnamefont {Bravyi}}, \
  and\ \bibinfo {author} {\bibfnamefont {J.~M.}\ \bibnamefont {Gambetta}},\
  }\href {\doibase 10.1103/PhysRevLett.119.180509} {\bibfield  {journal}
  {\bibinfo  {journal} {Phys. Rev. Lett.}\ }\textbf {\bibinfo {volume} {119}},\
  \bibinfo {pages} {180509} (\bibinfo {year} {2017})}\BibitemShut {NoStop}%
\bibitem [{\citenamefont {McClean}\ \emph {et~al.}(2016)\citenamefont
  {McClean}, \citenamefont {Romero}, \citenamefont {Babbush},\ and\
  \citenamefont {Aspuru-Guzik}}]{McClean.16}%
  \BibitemOpen
  \bibfield  {author} {\bibinfo {author} {\bibfnamefont {J.~R.}\ \bibnamefont
  {McClean}}, \bibinfo {author} {\bibfnamefont {J.}~\bibnamefont {Romero}},
  \bibinfo {author} {\bibfnamefont {R.}~\bibnamefont {Babbush}}, \ and\
  \bibinfo {author} {\bibfnamefont {A.}~\bibnamefont {Aspuru-Guzik}},\ }\href
  {https://iopscience.iop.org/article/10.1088/1367-2630/18/2/023023} {\bibfield
   {journal} {\bibinfo  {journal} {New J. Phys.}\ }\textbf {\bibinfo {volume}
  {18}},\ \bibinfo {pages} {023023} (\bibinfo {year} {2016})}\BibitemShut
  {NoStop}%
\bibitem [{\citenamefont {Ryabinkin}\ and\ \citenamefont
  {Genin}(2018)}]{Ryabinkin.18}%
  \BibitemOpen
  \bibfield  {author} {\bibinfo {author} {\bibfnamefont {I.~G.}\ \bibnamefont
  {Ryabinkin}}\ and\ \bibinfo {author} {\bibfnamefont {S.~N.}\ \bibnamefont
  {Genin}},\ }\href {https://arxiv.org/abs/1812.09812} {\bibfield  {journal}
  {\bibinfo  {journal} {arXiv:1812.09812}\ } (\bibinfo {year}
  {2018})}\BibitemShut {NoStop}%
\bibitem [{\citenamefont {McArdle}\ \emph {et~al.}(2019)\citenamefont
  {McArdle}, \citenamefont {Yuan},\ and\ \citenamefont
  {Benjamin}}]{SMcArdle.19}%
  \BibitemOpen
  \bibfield  {author} {\bibinfo {author} {\bibfnamefont {S.}~\bibnamefont
  {McArdle}}, \bibinfo {author} {\bibfnamefont {X.}~\bibnamefont {Yuan}}, \
  and\ \bibinfo {author} {\bibfnamefont {S.}~\bibnamefont {Benjamin}},\ }\href
  {https://link.aps.org/doi/10.1103/PhysRevLett.122.180501} {\bibfield
  {journal} {\bibinfo  {journal} {Phys. Rev. Lett.}\ }\textbf {\bibinfo
  {volume} {122}},\ \bibinfo {pages} {180501} (\bibinfo {year}
  {2019})}\BibitemShut {NoStop}%
\bibitem [{\citenamefont {Boixo}\ \emph {et~al.}(2018)\citenamefont {Boixo},
  \citenamefont {Isakov}, \citenamefont {Smelyanskiy}, \citenamefont {Babbush},
  \citenamefont {Ding}, \citenamefont {Jiang}, \citenamefont {Bremner},
  \citenamefont {Martinis},\ and\ \citenamefont {Neven}}]{Boixo.18}%
  \BibitemOpen
  \bibfield  {author} {\bibinfo {author} {\bibfnamefont {S.}~\bibnamefont
  {Boixo}}, \bibinfo {author} {\bibfnamefont {S.~V.}\ \bibnamefont {Isakov}},
  \bibinfo {author} {\bibfnamefont {V.~N.}\ \bibnamefont {Smelyanskiy}},
  \bibinfo {author} {\bibfnamefont {R.}~\bibnamefont {Babbush}}, \bibinfo
  {author} {\bibfnamefont {N.}~\bibnamefont {Ding}}, \bibinfo {author}
  {\bibfnamefont {Z.}~\bibnamefont {Jiang}}, \bibinfo {author} {\bibfnamefont
  {M.~J.}\ \bibnamefont {Bremner}}, \bibinfo {author} {\bibfnamefont {J.~M.}\
  \bibnamefont {Martinis}}, \ and\ \bibinfo {author} {\bibfnamefont
  {H.}~\bibnamefont {Neven}},\ }\href
  {https://www.nature.com/articles/s41567-018-0124-x} {\bibfield  {journal}
  {\bibinfo  {journal} {Nat. Phys.}\ }\textbf {\bibinfo {volume} {14}},\
  \bibinfo {pages} {595} (\bibinfo {year} {2018})}\BibitemShut {NoStop}%
\bibitem [{\citenamefont {Arute}\ \emph {et~al.}(2019)\citenamefont {Arute},
  \citenamefont {Arya}, \citenamefont {Babbush}, \citenamefont {Bacon},
  \citenamefont {Bardin}, \citenamefont {Barends}, \citenamefont {Biswas},
  \citenamefont {Boixo}, \citenamefont {Brandao}, \citenamefont {Buell} \emph
  {et~al.}}]{Arute.19}%
  \BibitemOpen
  \bibfield  {author} {\bibinfo {author} {\bibfnamefont {F.}~\bibnamefont
  {Arute}}, \bibinfo {author} {\bibfnamefont {K.}~\bibnamefont {Arya}},
  \bibinfo {author} {\bibfnamefont {R.}~\bibnamefont {Babbush}}, \bibinfo
  {author} {\bibfnamefont {D.}~\bibnamefont {Bacon}}, \bibinfo {author}
  {\bibfnamefont {J.~C.}\ \bibnamefont {Bardin}}, \bibinfo {author}
  {\bibfnamefont {R.}~\bibnamefont {Barends}}, \bibinfo {author} {\bibfnamefont
  {R.}~\bibnamefont {Biswas}}, \bibinfo {author} {\bibfnamefont
  {S.}~\bibnamefont {Boixo}}, \bibinfo {author} {\bibfnamefont {F.~G.}\
  \bibnamefont {Brandao}}, \bibinfo {author} {\bibfnamefont {D.~A.}\
  \bibnamefont {Buell}},  \emph {et~al.},\ }\href
  {https://www.nature.com/articles/s41586-019-1666-5#Sec9} {\bibfield
  {journal} {\bibinfo  {journal} {Nature (London)}\ }\textbf {\bibinfo {volume}
  {574}},\ \bibinfo {pages} {505} (\bibinfo {year} {2019})}\BibitemShut
  {NoStop}%
\bibitem [{\citenamefont {Ofek}\ \emph {et~al.}(2016)\citenamefont {Ofek},
  \citenamefont {Petrenko}, \citenamefont {Heeres}, \citenamefont {Reinhold},
  \citenamefont {Leghtas}, \citenamefont {Vlastakis}, \citenamefont {Liu},
  \citenamefont {Frunzio}, \citenamefont {Girvin}, \citenamefont {Jiang} \emph
  {et~al.}}]{Ofek.16}%
  \BibitemOpen
  \bibfield  {author} {\bibinfo {author} {\bibfnamefont {N.}~\bibnamefont
  {Ofek}}, \bibinfo {author} {\bibfnamefont {A.}~\bibnamefont {Petrenko}},
  \bibinfo {author} {\bibfnamefont {R.}~\bibnamefont {Heeres}}, \bibinfo
  {author} {\bibfnamefont {P.}~\bibnamefont {Reinhold}}, \bibinfo {author}
  {\bibfnamefont {Z.}~\bibnamefont {Leghtas}}, \bibinfo {author} {\bibfnamefont
  {B.}~\bibnamefont {Vlastakis}}, \bibinfo {author} {\bibfnamefont
  {Y.}~\bibnamefont {Liu}}, \bibinfo {author} {\bibfnamefont {L.}~\bibnamefont
  {Frunzio}}, \bibinfo {author} {\bibfnamefont {S.}~\bibnamefont {Girvin}},
  \bibinfo {author} {\bibfnamefont {L.}~\bibnamefont {Jiang}},  \emph
  {et~al.},\ }\href {https://www.nature.com/articles/nature18949} {\bibfield
  {journal} {\bibinfo  {journal} {Nature (London)}\ }\textbf {\bibinfo {volume}
  {536}},\ \bibinfo {pages} {441} (\bibinfo {year} {2016})}\BibitemShut
  {NoStop}%
\bibitem [{sm()}]{sm}%
  \BibitemOpen
  \href@noop {} {}\bibinfo {note} {See Supplemental Material for more
  numerical results, theoretical analysis and discussions.}\BibitemShut {Stop}%
\bibitem [{\citenamefont {Parashar}\ and\ \citenamefont
  {Rana}(2009)}]{Parashar.09}%
  \BibitemOpen
  \bibfield  {author} {\bibinfo {author} {\bibfnamefont {P.}~\bibnamefont
  {Parashar}}\ and\ \bibinfo {author} {\bibfnamefont {S.}~\bibnamefont
  {Rana}},\ }\href {https://link.aps.org/doi/10.1103/PhysRevA.80.012319}
  {\bibfield  {journal} {\bibinfo  {journal} {Phys. Rev. A}\ }\textbf {\bibinfo
  {volume} {80}},\ \bibinfo {pages} {012319} (\bibinfo {year}
  {2009})}\BibitemShut {NoStop}%
\bibitem [{\citenamefont {Christandl}\ \emph {et~al.}(2004)\citenamefont
  {Christandl}, \citenamefont {Datta}, \citenamefont {Ekert},\ and\
  \citenamefont {Landahl}}]{Christandl.04}%
  \BibitemOpen
  \bibfield  {author} {\bibinfo {author} {\bibfnamefont {M.}~\bibnamefont
  {Christandl}}, \bibinfo {author} {\bibfnamefont {N.}~\bibnamefont {Datta}},
  \bibinfo {author} {\bibfnamefont {A.}~\bibnamefont {Ekert}}, \ and\ \bibinfo
  {author} {\bibfnamefont {A.~J.}\ \bibnamefont {Landahl}},\ }\href
  {https://link.aps.org/doi/10.1103/PhysRevLett.92.187902} {\bibfield
  {journal} {\bibinfo  {journal} {Phys. Rev. Lett.}\ }\textbf {\bibinfo
  {volume} {92}},\ \bibinfo {pages} {187902} (\bibinfo {year}
  {2004})}\BibitemShut {NoStop}%
\bibitem [{\citenamefont {Mitarai}\ \emph {et~al.}(2018)\citenamefont
  {Mitarai}, \citenamefont {Negoro}, \citenamefont {Kitagawa},\ and\
  \citenamefont {Fujii}}]{Mitarai.18}%
  \BibitemOpen
  \bibfield  {author} {\bibinfo {author} {\bibfnamefont {K.}~\bibnamefont
  {Mitarai}}, \bibinfo {author} {\bibfnamefont {M.}~\bibnamefont {Negoro}},
  \bibinfo {author} {\bibfnamefont {M.}~\bibnamefont {Kitagawa}}, \ and\
  \bibinfo {author} {\bibfnamefont {K.}~\bibnamefont {Fujii}},\ }\href
  {https://journals.aps.org/pra/abstract/10.1103/PhysRevA.98.032309} {\bibfield
   {journal} {\bibinfo  {journal} {Phy. Rev. A}\ }\textbf {\bibinfo {volume}
  {98}},\ \bibinfo {pages} {032309} (\bibinfo {year} {2018})}\BibitemShut
  {NoStop}%
\bibitem [{\citenamefont {Johnson}\ \emph {et~al.}(2017)\citenamefont
  {Johnson}, \citenamefont {Romero}, \citenamefont {Olson}, \citenamefont
  {Cao},\ and\ \citenamefont {Aspuru-Guzik}}]{Johnson.17}%
  \BibitemOpen
  \bibfield  {author} {\bibinfo {author} {\bibfnamefont {P.~D.}\ \bibnamefont
  {Johnson}}, \bibinfo {author} {\bibfnamefont {J.}~\bibnamefont {Romero}},
  \bibinfo {author} {\bibfnamefont {J.}~\bibnamefont {Olson}}, \bibinfo
  {author} {\bibfnamefont {Y.}~\bibnamefont {Cao}}, \ and\ \bibinfo {author}
  {\bibfnamefont {A.}~\bibnamefont {Aspuru-Guzik}},\ }\href
  {https://arxiv.org/pdf/1711.02249.pdf} {\bibfield  {journal} {\bibinfo
  {journal} {arXiv:1711.02249}\ } (\bibinfo {year} {2017})}\BibitemShut
  {NoStop}%
\bibitem [{\citenamefont {Strikis}\ \emph {et~al.}(2020)\citenamefont
  {Strikis}, \citenamefont {Qin}, \citenamefont {Chen}, \citenamefont
  {Benjamin},\ and\ \citenamefont {Li}}]{Strikis.20}%
  \BibitemOpen
  \bibfield  {author} {\bibinfo {author} {\bibfnamefont {A.}~\bibnamefont
  {Strikis}}, \bibinfo {author} {\bibfnamefont {D.}~\bibnamefont {Qin}},
  \bibinfo {author} {\bibfnamefont {Y.}~\bibnamefont {Chen}}, \bibinfo {author}
  {\bibfnamefont {S.~C.}\ \bibnamefont {Benjamin}}, \ and\ \bibinfo {author}
  {\bibfnamefont {Y.}~\bibnamefont {Li}},\ }\href
  {https://arxiv.org/abs/2005.07601} {\bibfield  {journal} {\bibinfo  {journal}
  {arXiv:2005.07601}\ } (\bibinfo {year} {2020})}\BibitemShut {NoStop}%
\bibitem [{\citenamefont {Carolan}\ \emph {et~al.}(2015)\citenamefont
  {Carolan}, \citenamefont {Harrold}, \citenamefont {Sparrow}, \citenamefont
  {Mart{\'\i}n-L{\'o}pez}, \citenamefont {Russell}, \citenamefont
  {Silverstone}, \citenamefont {Shadbolt}, \citenamefont {Matsuda},
  \citenamefont {Oguma}, \citenamefont {Itoh} \emph {et~al.}}]{Carolan.15}%
  \BibitemOpen
  \bibfield  {author} {\bibinfo {author} {\bibfnamefont {J.}~\bibnamefont
  {Carolan}}, \bibinfo {author} {\bibfnamefont {C.}~\bibnamefont {Harrold}},
  \bibinfo {author} {\bibfnamefont {C.}~\bibnamefont {Sparrow}}, \bibinfo
  {author} {\bibfnamefont {E.}~\bibnamefont {Mart{\'\i}n-L{\'o}pez}}, \bibinfo
  {author} {\bibfnamefont {N.~J.}\ \bibnamefont {Russell}}, \bibinfo {author}
  {\bibfnamefont {J.~W.}\ \bibnamefont {Silverstone}}, \bibinfo {author}
  {\bibfnamefont {P.~J.}\ \bibnamefont {Shadbolt}}, \bibinfo {author}
  {\bibfnamefont {N.}~\bibnamefont {Matsuda}}, \bibinfo {author} {\bibfnamefont
  {M.}~\bibnamefont {Oguma}}, \bibinfo {author} {\bibfnamefont
  {M.}~\bibnamefont {Itoh}},  \emph {et~al.},\ }\href
  {https://dx.doi.org/10.1126/science.aab3642} {\bibfield  {journal} {\bibinfo
  {journal} {Science}\ }\textbf {\bibinfo {volume} {349}},\ \bibinfo {pages}
  {711} (\bibinfo {year} {2015})}\BibitemShut {NoStop}%
\bibitem [{Rum()}]{Rumelhart.85}%
  \BibitemOpen
  \href@noop {} {}\bibinfo {note} {D. E. Rumelhart, G. E. Hinton, and R. J.
  Williams,
  \href{https://web.stanford.edu/class/psych209a/ReadingsByDate/02_06/PDPVolIChapter8.pdf}{\textit{Learning
  internal representations by error propagation}} \rm{in}
  \href{https://dl.acm.org/doi/book/10.5555/104279}{\textit{Parallel
  distributed processing: explorations in the microstructure of cognition, Vol.
  1: foundations}}, edited by D. E. Rumelhart and J. L McClelland (MIT Press,
  Cambridge, 1985) pp. 318-362}\BibitemShut {NoStop}%
\bibitem [{\citenamefont {Kingma}\ and\ \citenamefont {Ba}(2014)}]{Kingma.14}%
  \BibitemOpen
  \bibfield  {author} {\bibinfo {author} {\bibfnamefont {D.~P.}\ \bibnamefont
  {Kingma}}\ and\ \bibinfo {author} {\bibfnamefont {J.}~\bibnamefont {Ba}},\
  }\href {https://arxiv.org/abs/1412.6980} {\bibfield  {journal} {\bibinfo
  {journal} {arXiv:1412.6980}\ } (\bibinfo {year} {2014})}\BibitemShut
  {NoStop}%
\bibitem [{\citenamefont {Kirkpatrick}\ \emph {et~al.}(1983)\citenamefont
  {Kirkpatrick}, \citenamefont {Gelatt},\ and\ \citenamefont
  {Vecchi}}]{Kirkpatrick.83}%
  \BibitemOpen
  \bibfield  {author} {\bibinfo {author} {\bibfnamefont {S.}~\bibnamefont
  {Kirkpatrick}}, \bibinfo {author} {\bibfnamefont {C.~D.}\ \bibnamefont
  {Gelatt}}, \ and\ \bibinfo {author} {\bibfnamefont {M.~P.}\ \bibnamefont
  {Vecchi}},\ }\href {https://science.sciencemag.org/content/220/4598/671}
  {\bibfield  {journal} {\bibinfo  {journal} {Science}\ }\textbf {\bibinfo
  {volume} {220}},\ \bibinfo {pages} {671} (\bibinfo {year}
  {1983})}\BibitemShut {NoStop}%
\bibitem [{\citenamefont {Niu}\ \emph {et~al.}(2019)\citenamefont {Niu},
  \citenamefont {Boixo}, \citenamefont {Smelyanskiy},\ and\ \citenamefont
  {Neven}}]{Niu.19}%
  \BibitemOpen
  \bibfield  {author} {\bibinfo {author} {\bibfnamefont {M.~Y.}\ \bibnamefont
  {Niu}}, \bibinfo {author} {\bibfnamefont {S.}~\bibnamefont {Boixo}}, \bibinfo
  {author} {\bibfnamefont {V.~N.}\ \bibnamefont {Smelyanskiy}}, \ and\ \bibinfo
  {author} {\bibfnamefont {H.}~\bibnamefont {Neven}},\ }\href
  {https://www.nature.com/articles/s41534-019-0141-3} {\bibfield  {journal}
  {\bibinfo  {journal} {npj Quantum Inf.}\ }\textbf {\bibinfo {volume} {5}},\
  \bibinfo {pages} {33} (\bibinfo {year} {2019})}\BibitemShut {NoStop}%
\bibitem [{\citenamefont {Xu}\ \emph {et~al.}(2019)\citenamefont {Xu},
  \citenamefont {Li}, \citenamefont {Liu}, \citenamefont {Wang}, \citenamefont
  {Yuan},\ and\ \citenamefont {Wang}}]{Xu.19}%
  \BibitemOpen
  \bibfield  {author} {\bibinfo {author} {\bibfnamefont {H.}~\bibnamefont
  {Xu}}, \bibinfo {author} {\bibfnamefont {J.}~\bibnamefont {Li}}, \bibinfo
  {author} {\bibfnamefont {L.}~\bibnamefont {Liu}}, \bibinfo {author}
  {\bibfnamefont {Y.}~\bibnamefont {Wang}}, \bibinfo {author} {\bibfnamefont
  {H.}~\bibnamefont {Yuan}}, \ and\ \bibinfo {author} {\bibfnamefont
  {X.}~\bibnamefont {Wang}},\ }\href
  {https://www.nature.com/articles/s41534-019-0198-z} {\bibfield  {journal}
  {\bibinfo  {journal} {npj Quantum Inf.}\ }\textbf {\bibinfo {volume} {5}},\
  \bibinfo {pages} {82} (\bibinfo {year} {2019})}\BibitemShut {NoStop}%
\bibitem [{\citenamefont {Zhang}\ \emph {et~al.}(2019)\citenamefont {Zhang},
  \citenamefont {Wei}, \citenamefont {Asad}, \citenamefont {Yang},\ and\
  \citenamefont {Wang}}]{Zhang.19}%
  \BibitemOpen
  \bibfield  {author} {\bibinfo {author} {\bibfnamefont {X.-M.}\ \bibnamefont
  {Zhang}}, \bibinfo {author} {\bibfnamefont {Z.}~\bibnamefont {Wei}}, \bibinfo
  {author} {\bibfnamefont {R.}~\bibnamefont {Asad}}, \bibinfo {author}
  {\bibfnamefont {X.-C.}\ \bibnamefont {Yang}}, \ and\ \bibinfo {author}
  {\bibfnamefont {X.}~\bibnamefont {Wang}},\ }\href
  {https://www.nature.com/articles/s41534-019-0201-8} {\bibfield  {journal}
  {\bibinfo  {journal} {npj Quantum Inf.}\ }\textbf {\bibinfo {volume} {5}},\
  \bibinfo {pages} {85} (\bibinfo {year} {2019})}\BibitemShut {NoStop}%
\bibitem [{\citenamefont {Lloyd}(1996)}]{Lloyd.96}%
  \BibitemOpen
  \bibfield  {author} {\bibinfo {author} {\bibfnamefont {S.}~\bibnamefont
  {Lloyd}},\ }\href {https://science.sciencemag.org/content/273/5278/1073}
  {\bibfield  {journal} {\bibinfo  {journal} {Science}\ }\textbf {\bibinfo
  {volume} {273}},\ \bibinfo {pages} {1073} (\bibinfo {year}
  {1996})}\BibitemShut {NoStop}%
\bibitem [{\citenamefont {Berry}\ \emph {et~al.}(2007)\citenamefont {Berry},
  \citenamefont {Ahokas}, \citenamefont {Cleve},\ and\ \citenamefont
  {Sanders}}]{Berry.07}%
  \BibitemOpen
  \bibfield  {author} {\bibinfo {author} {\bibfnamefont {D.~W.}\ \bibnamefont
  {Berry}}, \bibinfo {author} {\bibfnamefont {G.}~\bibnamefont {Ahokas}},
  \bibinfo {author} {\bibfnamefont {R.}~\bibnamefont {Cleve}}, \ and\ \bibinfo
  {author} {\bibfnamefont {B.~C.}\ \bibnamefont {Sanders}},\ }\href
  {https://link.springer.com/article/10.1007/s00220-006-0150-x} {\bibfield
  {journal} {\bibinfo  {journal} {Commun. Math. Phys.}\ }\textbf {\bibinfo
  {volume} {270}},\ \bibinfo {pages} {359} (\bibinfo {year}
  {2007})}\BibitemShut {NoStop}%
\bibitem [{\citenamefont {Weedbrook}\ \emph {et~al.}(2012)\citenamefont
  {Weedbrook}, \citenamefont {Pirandola}, \citenamefont
  {Garc{\'\i}a-Patr{\'o}n}, \citenamefont {Cerf}, \citenamefont {Ralph},
  \citenamefont {Shapiro},\ and\ \citenamefont {Lloyd}}]{Weedbrook.12}%
  \BibitemOpen
  \bibfield  {author} {\bibinfo {author} {\bibfnamefont {C.}~\bibnamefont
  {Weedbrook}}, \bibinfo {author} {\bibfnamefont {S.}~\bibnamefont
  {Pirandola}}, \bibinfo {author} {\bibfnamefont {R.}~\bibnamefont
  {Garc{\'\i}a-Patr{\'o}n}}, \bibinfo {author} {\bibfnamefont {N.~J.}\
  \bibnamefont {Cerf}}, \bibinfo {author} {\bibfnamefont {T.~C.}\ \bibnamefont
  {Ralph}}, \bibinfo {author} {\bibfnamefont {J.~H.}\ \bibnamefont {Shapiro}},
  \ and\ \bibinfo {author} {\bibfnamefont {S.}~\bibnamefont {Lloyd}},\ }\href
  {https://journals.aps.org/rmp/abstract/10.1103/RevModPhys.84.621} {\bibfield
  {journal} {\bibinfo  {journal} {Rev. Mod. Phys.}\ }\textbf {\bibinfo {volume}
  {84}},\ \bibinfo {pages} {621} (\bibinfo {year} {2012})}\BibitemShut
  {NoStop}%
\end{thebibliography}
\end{document}